\documentclass[12pt,draftcls,onecolumn,peerreview]{IEEEtran}

\makeatletter
\usepackage{epsfig, psfrag}
\usepackage{subfigure}
\usepackage{epstopdf}
\usepackage{psfrag}
\usepackage{latexsym}
\usepackage{amssymb}
\usepackage{amsthm} 
\usepackage{placeins}
\usepackage{graphics,color}
\usepackage{graphicx}
\usepackage{amsmath}
\usepackage{cite}
\usepackage{comment}
\usepackage{url}
\usepackage{xspace}

\usepackage[top= 1 in, bottom= 1 in, left= 1 in, right=1 in]{geometry}

\newcommand\T{\rule{0pt}{2.6ex}}       

\def\Pe{\mathrm {Pe}}

 \DeclareMathOperator*{\argmax}{arg\,max}
 \DeclareMathOperator*{\argmin}{arg\,min}
\newtheorem{theorem}{Theorem}
\newtheorem{lemma}{Lemma}
\newtheorem{proposition}{Proposition}
\newtheorem{corollary}{Corollary}

\theoremstyle{definition}
\newtheorem{definition} {Definition}
\newtheorem{remarks}{Remark}
\newtheorem{example}{Example}

\newfont{\boldlarge}{msbm10 scaled 1100}

\newcommand{\ignore}[1]{}

\def\expe{\mathbb{E}}   
\def\argmin{\mathop{\rm argmin}}
\def\argmax{\mathop{\rm argmax}}

\def\P{\mathsf{P}}
\def\ber{\mathsf{Bern}}

\newcommand{\kl}[2]{D\left(  #1 \left\| #2 \right. \right)}

\def\mbf{\mathbf}
\def\mbs{\boldsymbol}

\renewcommand{\qed}{\nobreak \ifvmode \relax \else
      \ifdim\lastskip<1.5em \hskip-\lastskip
      \hskip1.5em plus0em minus0.5em \fi \nobreak
      \vrule height0.2em width0.5em depth0.4em\fi}
\IEEEoverridecommandlockouts
\begin{document}

\sloppy

\title{Improved Target Acquisition Rates with Feedback Codes}

\author{Anusha~Lalitha, 
        Nancy~Ronquillo, 
        and~Tara~Javidi,~\IEEEmembership{Senior~Member,~IEEE}
\thanks{Preliminary versions of this work were presented at the 50th Asilomar Conference on Signals, Systems, and Computers, and at 2017 International Symposium on Information Theory.}
\thanks{A. Lalitha N. Ronquillo, and T. Javidi are with the Department of Electrical and Computer Engineering, University of California San Diego, La Jolla, CA 92093, USA. (e-mail: alalitha@ucsd.edu; nronquil@ucsd,edu; tjavidi@ucsd.edu). } 
}

\maketitle
\begin{abstract}
This paper considers the problem of acquiring an unknown target location (among a finite number of locations) via a sequence of measurements, where each measurement consists of simultaneously probing a group of locations. The resulting observation consists of a sum of an indicator of the target's presence in the probed region, and a zero mean Gaussian noise term whose variance is a function of the measurement vector. An equivalence between the target acquisition problem and channel coding over a binary input additive white Gaussian noise (BAWGN) channel with state and feedback is established. Utilizing this information theoretic perspective, a two-stage adaptive target search strategy based on the sorted Posterior Matching channel coding strategy is proposed. Furthermore, using information theoretic converses, the fundamental limits on the target acquisition rate for adaptive and non-adaptive strategies are characterized. As a corollary to the non-asymptotic upper bound of the expected number of measurements under the proposed two-stage strategy, and to non-asymptotic lower bound of the expected number of measurements for optimal non-adaptive search strategy, a lower bound on the adaptivity gain is obtained. The adaptivity gain is further investigated in different asymptotic regimes of interest.
\end{abstract}


\section{Introduction}
\label{sec:intro}
Consider a single target acquisition over a search region of width $B$ and resolution up to width $\delta$. Mathematically, this is the problem of estimating a unit vector $\mbf{W} \in \{0,1\}^{\frac{B}{\delta}}$ via a sequence of noisy linear measurements 
\begin{equation}
\label{eq:linear1}
Y_n = \langle {\mbf{S}_n},\mbf{W}  + \mbs{\Xi}_n\rangle, \quad n = 1,2, \ldots, \tau,
\end{equation}
where a binary measurement vector $\mbf{S}_n \in \{0,1\}^{\frac{B}{\delta}}$ denotes the locations inspected and the vector $\mbs{\Xi}_n \in \mathbb{R}^{\frac{B}{\delta}}$ denotes the additive measurement noise per location. More generally, the observation $Y_n$ at time $n$ can be written as 
\begin{equation}
\label{eq:linear2}
Y_n = \langle\mbf{S}_n, \mbf{W}\rangle + Z_n(\mbf{S}_n),
\end{equation}
where $Z_n(\mbf{S}_n)$ is a noise term whose statistics are a function of the measurement vector $\mbf{S}_n$. The goal is to design the sequence of measurement vectors $\{\mbf{S}_n\}_{n = 1}^{\tau}$, such that the target location $\textbf{W}$ is estimated with high reliability, while keeping the (expected) number of measurements $\tau$ as low as possible.

In this paper, we first consider the linear model~\eqref{eq:linear1} when the elements of $\mbf{\Xi}_n$ are i.i.d Gaussian with zero mean and variance $\delta \sigma^2$. This means that $Z_n(\mbf{S}_n)$ in~\eqref{eq:linear2} are distributed as $\mathcal{N}(0, |\mbf{S}_n| \delta \sigma^2)$, and show that the problem of searching for a target under measurement dependent Gaussian noise $Z_n(\mbf{S}_n)$ is equivalent to channel coding over a binary additive white Gaussian noise (BAWGN) channel with state and feedback (in Section 4.6~\cite{Gallager}). This allows us not only to retrofit the known channel coding schemes based on sorted Posterior Matching (sort PM)~\cite{SungEnChiu} as adaptive search strategies, but also to obtain information theoretic converses to characterize fundamental limits on the target acquisition rate under both adaptive and non-adaptive strategies. As a corollary to the non-asymptotic analysis of our sorted Posterior-Matching-based adaptive strategy and our converse for non-adaptive strategy, we obtain a lower bound on the adaptivity gain.

\subsection{Our Contributions}

Our main results are inspired by the analogy between target acquisition under measurement dependent noise and channel coding with state and feedback. This connection was utilized in~\cite{DBLP:journals/corr/KaspiSJ16} under a Bernoulli noise model. In this paper, in Proposition~\ref{prop:connection}, we formalize the connection between our target acquisition problem with Gaussian measurement dependent noise and channel coding over a BAWGN channel with state. Here, the  channel state denotes the variance of the measurement dependent noise $ |\mbf{S}_n| \delta \sigma^2$. Since feedback codes i.e., adapting the codeword to the past channel outputs, are known to increase the capacity of a channel with state and feedback. This motivates us to use adaptivity when searching, i.e., to utilize past observations $\{Y_1, Y_2, \ldots, Y_{n-1}\}$ when selecting the next measurement vector $\mbf{S}_n$.  Furthermore, this information theoretic perspective allows us to quantify the increase in the adaptive target acquisition rate. Our analysis of improvement in the target acquisition rate as well as the adaptivity gain, measured as the reduction in expected number of measurements, while using an adaptive strategy over a non-adaptive strategy has two components. Firstly, we utilize information theoretic converse for an optimal non-adaptive search strategy to obtain a non-asymptotic lower bound on the minimum expected number of measurements required while maintaining a desired reliability. As a consequence, this provides the best non-adaptive target acquisition rate. Secondly, we utilize a feedback code based on Posterior Matching as a two-stage adaptive search strategy and obtain a non-asymptotic upper bound on the expected number of measurements while maintaining a desired reliability. These two components of our analysis allow us to characterize a lower bound on the increased target acquisition rate due to adaptivity. 

Our non-asymptotic analysis of adaptivity gain reveals two qualitatively different asymptotic regimes. In particular, we show that adaptivity gain depends on the manner in which the number of locations grow. We show that the adaptivity grows logarithmically in the number of locations, i.e., $O\left(\log \frac{B}{\delta} \right)$ when refining the search resolution $\delta$ ($\delta$ going to zero) and while keeping total search width $B$ fixed. On the other hand, we show that as the search width $B$ expands while keeping search resolution $\delta$ fixed, the adaptivity gain grows in the number of locations as $O\left(\frac{B}{\delta} \log \frac{B}{\delta} \right)$.


The problem of searching for a target under a binary measurement dependent noise, whose crossover probability increases with the weight of the measurement vector was studied by~\cite{DBLP:journals/corr/KaspiSJ16} and analyzed under sort PM strategy in~\cite{SungEnChiu}. In particular, \cite{DBLP:journals/corr/KaspiSJ16} and~\cite{SungEnChiu} provide asymptotic analysis of the adaptivity gain for the case where $B = 1$ and $\delta $ approaches zero. Our prior work~\cite{8007098} by utilizing a (suboptimal) hard decoding of Gaussian observation $Y_n$, strengthens~\cite{DBLP:journals/corr/KaspiSJ16} and~\cite{SungEnChiu} by also accounting for the regime in which $B$ grows. While the analysis in~\cite{8007098} strengthens the non-asymptotic bounds in~\cite{SungEnChiu} with Bernoulli noise it failed to provide tight analysis for our problem with Gaussian observations. In this paper, by strengthening our analysis in~\cite{8007098} we extend the prior work in three ways: (i) we consider the soft Gaussian observation $Y_n$, (ii) we obtain non-asymptotic achievability and converse analysis, and (iii) we characterize tight non-asymptotic adaptivity gain in the two asymptotically distinct regimes of $B \to \infty$ and $\delta \to 0$.

\subsection{Applications}
Our problem formulation addresses two challenging engineering problems which arise in the context of modern communication systems. We will discuss the two problems in the following examples and then provide the details of the state of art.

\begin{example}[Establishing initial access in mm-Wave communication] Consider the problem of detecting the direction of arrival for initial access in millimeter wave (mmWave) Communications. In mmWave communication, prior to data transmission the base station is tasked with aligning the transmitter and receiver antennas in the angular space. In other words, the base station's antenna pattern can be viewed as a measurement vector $\mbf{S}_n$ searching the angular space $B \subset (0, 360^{\circ})$. At each time $n$, the noise intensity depends on the base station's antenna pattern $\mbf{S}_n$ and the noisy observation $Y_n$ is a function of measurement dependent noise $Z_n(\mbf{S}_n)$. Here it is natural to characterize the fundamental limit on the measurement time as a function of asymptotically small~$\delta$. 
\end{example}

\begin{example}[Spectrum Sensing for Cognitive Radio] 
Consider the problem of opportunistically searching for a vacant subband of bandwidth $\delta$ over a total bandwidth of $B$. In this problem secondary user desires to locate the single stationary vacant subband quickly and reliably, by making measurements $\mbf{S}_n$ at every time $n$. At each time instant $n$, the noise intensity depends on the number of subbands probed as dictated by $\mbf{S}_n$ and noisy observation $Y_n$ is is a function of measurement dependent noise $Z_n(\mbf{S}_n)$. Here it is natural to characterize fundamental limit of the measurement time required for a secondary user to acquire the vacant subband as a function of the asymptotically large bandwidth $B$.
\end{example}

Giordani~et~al.~\cite{7460513} compare the exhaustive search like the Sequential Beamspace Scanning considered by Barati~et~al.~\cite{7421136}, where the base station sequentially searches through all angular sectors, against a two stage iterative hierarchical search strategy. In the first stage an exhaustive search identifies a coarse sector by repeatedly probing each coarse region for a predetermined SNR to be achieved. In the second stage an exhaustive search over all locations identifies the target. Giordani~et~al.~show that in general the adaptive iterative strategy reduces the number of measurements over exhaustive search except when desired SNR is too high, forcing the number of measurements required at each stage to get too large. We observe this in through our simulations in Section~\ref{sec:num_results}-A. In fact, as confirmed by our simulations random-coding-based non-adaptive strategies including the Agile-Link protocol~\cite{Abari_AgileLink}, outperform the repetition based adaptive strategies.

Past literature on spectrum sensing for cognitive radio~\cite{42_35Multibandjoint, 50AdaptiveMultiband, 55AdaptiveAgileCR} and support vector recovery~\cite{Nowak_CompSensing,Y_Kim_MACSensing} have focused on the problem where $\textbf{S}_n$ can be real or complex, with measurement independent noise applying both exhaustive search and multiple adaptive search strategies. In contrast, our work considers a simple binary model, $\textbf{S}_n \in \{0,1\}^{\frac{B}{\delta}}$, but captures the implications of measurement dependence of the noise, which is known in the spectrum sensing literature as noise folding. The problem of measurement dependent noise (known as noise folding) has been investigated in~\cite{Treichler_NoiseFolding} where non-adaptive design of complex measurements matrix satisfying RIP condition has been investigated.  Our work compliments this study by characterizing the gain associated with adaptively addressing the measurement dependent noise (noise folding), albeit for the simpler case of binary measurements. We note that the case of adptively finding a subset of a sufficiently large vacant bandwidth with noise folding is considered in~\cite{Sharma_Murthy}, where ideas from group testing and noisy binary search have been utilized. The solutions however depend strongly on  the availability of sufficiently large consective vacant band and does not apply to our setting.



\noindent\underline{Notations:} Vectors are denoted by boldface letters $\mbf{A}$ and $\mbf{A}{(j)}$ is the $j^{th}$ element of a vector. Matrices are denoted by overlined boldface letters.  Let $\mathcal{U}_M$ denote the set $\{\textbf{u}\in \mathbb{R}^M: u(j)\in\{0,1\}  \}$. Bern$(p)$ denotes the Bernoulli distribution with parameter $p$, $h(p) = - p\log p -(1-p)\log(1-p)$ denotes the entropy of a Bernoulli random variable with parameter $p$. Let $G(x; \mu, \sigma^2)$ denote the pdf of Gaussian random variable with mean $\mu$ and variance $\sigma^2$ at $x$. Logarithms are to the base 2. Let $[g]_a = g$ if $g \geq a$ otherwise $[g]_a = 0$.


\section{Problem Setup}
\label{sec:prob_setup}
In this section, we describe the mathematical formulation of the target acquisition problem followed by the performance criteria.

\subsection{Problem Formulation}
We consider a search agent interested in quickly and reliably finding the true location of a single stationary target by making measurements over time about the target's presence. In particular, we consider a total search region of width $B$ that contains the target in a location of width $\delta$. In other words, the search agent is searching for the target's location among $ \frac{B}{\delta}$ total locations. Let $\mbf{W} \in \mathcal{U}_{\frac{B}{\delta}}$ denote the true location of the target, where $\mbf{W}(j)=1$ if and only if target is located at location $j$. The target location $\mbf{W}$ can take $\frac{B}{\delta}$ possible values uniformly at random whose value remains fixed during the search. A measurement at time $n$ is given by a vector $\mbf{S}_n \in \mathcal{U}_{\frac{B}{\delta}}$, where $\mbf{S}_n(j)=1$ if and only if location $j$ is probed. Each measurement can be imagined to result in a clean observation $X_{n} =  \textbf{W}^{\intercal} \mbf{S}_{n} \in \{0,1 \}$ indicating of the presence of the target in the measurement vector $\textbf{S}_n$. However, only a noisy version of the clean observation $X_n$ is available to the agent.
The resulting noisy observation  $Y_n \in \mathbb{R}$ is given by the following linear model with additive measurement dependent noise
\begin{equation}
Y_n = X_{n}+  {Z}_{n}(\mbf{S}_n).
\label{eq:noisysearch}
\end{equation}
Here, we assume ${Z}_{n} \sim \mathcal{N}(0, |\textbf{S}_n|\delta \sigma^2)$ which corresponds to the case of i.i.d white Gaussian noise with $\sigma^2$ denotes the noise variance per unit width. Conditioned on the measurement vector $\mbf{S}_n$, the noise $Z_{n}$ is independent over time. 

A search consisting of $\tau$ measurements can be represented by a measurement matrix $\overline{\mbf{S}}^{\tau} = [\mbf{S}_1, \mbf{S}_2, \ldots , \mbf{S}_{\tau}]$ which yields the observation vector $\mbf{Y}^{\tau} = [Y_1, Y_2, \ldots, Y_{\tau}]$. At any time instant $n = 1,2 \ldots, \tau$, the agent selects the measurement vector in general as a function of the past observations and measurements. Mathematically,
\begin{align}
\mbf{S}_{n} = g_{n}\left(\mbf{Y}^{n-1}, \overline{\mbf{S}}^{n-1}\right),
\end{align} 
for some causal (possibly random) function $g_{n}: \mathbb{R}^{n-1} \times \mathcal{U}^{n-1}_{\frac{B}{\delta}}  \to \mathcal{U}_{\frac{B}{\delta}}$. After observing the noisy observations $\mbf{Y}^{\tau}$ and measurement matrix $\overline{\mbf{S}}^{\tau}$, the agent estimates the target location $\mbf{W}$ as follows
\begin{align}
\hat{\mbf{W}} = d\left( \mbf{Y}^{\tau}, \overline{\mbf{S}}^{\tau}\right),
\end{align}
for some decision function $d: \mathbb{R}^{\tau} \times \mathcal{U}^{\tau}_{\frac{B}{\delta}}  \to \mathcal{U}_{\frac{B}{\delta}}$. The probability of error  for a search is given by $\Pe = \P(\hat{\mbf{W}} \neq \mbf{W} | \mbf{Y}, \overline{\mbf{S}})$ and the average probability of error  is given by $\overline{\Pe} = \P(\hat{\mbf{W}} \neq \mbf{W})$.

Now we define the measurement strategy:

\begin{definition}[$\epsilon$-Reliable Search Strategy $\mathfrak{c}_{\epsilon}$]
For some $\epsilon \in(0, 1)$, an \textit{$\epsilon$-reliable search strategy}, denoted by $\mathfrak{c}_{\epsilon}$, is defined as a sequence of $\tau$ (possibly random) number of
causal functions $\{g_1, g_2, \ldots, g_{\tau}\}$, according to which the measurement matrix $\overline{\mbf{S}}^{\tau}$ is selected, and a decision function $d$ which provides an estimate $\mbf{\hat{W}}$ of $\mbf{W}$, such that the average probability of error $\overline{\Pe}$ is at most $\epsilon$.
\end{definition}

\begin{definition}[Achievable Target Acquisition Rate]
A target acquisition rate $R$ is said to be an \textit{$\epsilon$-achievable}, if for any small $\xi > 0$ and $n$ large enough, there exists an $\epsilon$-reliable search strategy $\mathfrak{c}_{\epsilon}$ such the following holds
\begin{align}
\expe_{\mathfrak{c}_{\epsilon}}[\tau] &\leq n, \\
\frac{B}{\delta} &\geq 2^{n(R-\xi)}.
\end{align}
A targeting rate $R$ is said to be \textit{achievable target acquisition rate} if it is $\epsilon$-achievable for all $\epsilon \in (0,1)$. 
\end{definition}
The above definition is motivated by information theoretic notion of transmission rate over a communication channel, which captures the exponential rate at which the number of messages grow with the number of channel uses while the receiver can decode with a small average error probability. Similarly, the target acquisition rate captures the exponential rate at which the number of target locations grow with the number of measurement vectors while a search strategy can still locate the target with a diminishing average error probability.

\begin{definition}[Target Acquisition Capacity]
The supremum of achievable target acquisition rates is called the target acquisition capacity.
\end{definition}

\subsection{Types of Search Strategies and Adaptivity Gain}

Each measurement vector $\mbf{S}_n$ and the number of total measurements $\tau$ can be selected either based on the past observations $\mbf{Y}^{n-1}$, or independent of them. Based on these two choices, strategies can be divided into four types i) having fixed length versus variable length number of the measurement matrix $\overline{\mbf{S}}$, and ii) being adaptive versus non-adaptive.   
A \textit{fixed length $\epsilon$-reliable strategy} $\mathfrak{c}_{\epsilon}$ uses a fixed number of measurements $\tau$ predetermined offline independent of the observations, to obtain estimate $\hat{\mbf{W}}$. On the other hand, a \textit{variable length $\epsilon$-reliable strategy} $\mathfrak{c}_{\epsilon}$ uses a random number of measurements $\tau$ (possibly determined as a function of the observations $\mbf{Y}^{\tau}$) to obtain estimate, $\hat{\mbf{W}}$. For example, $\tau$ can be selected such that agent achieves $\Pe \leq \epsilon$ in every search and hence $\tau$ is a random variable which is a function of the past noisy observations. Under an \textit{adaptive strategy} $\mathfrak{c}_{\epsilon} \in \mathcal{C}^A_{\epsilon}$ the agent designs the measurement vector $\textbf{S}_n$ as a function of the past observations $\mbf{Y}^{n-1}$, i.e., $g_n$ is a function of both $\textbf{S}^{n-1}$ and $\mbf{Y}^{n-1}$.

\begin{definition}
Let $\mathcal{C}^A_{\epsilon}$ be a class of all $\epsilon$-reliable adaptive strategies. 
\end{definition}

Under a \textit{non-adaptive strategy}, the agent designs the measurement vector $\textbf{S}_n$ offline independent of past observations, i.e., $g_n$ does not depend on $\textbf{S}^{n-1}$ or $\mbf{Y}^{n-1}$. 

\begin{definition}
Let $\mathcal{C}^{NA}_{\epsilon}$ be a class of all $\epsilon$-reliable non-adaptive strategies. 
\end{definition}

For any $\epsilon$-reliable strategy $\mathfrak{c}_{\epsilon}$, the performance is measured by the expected number of measurements $\expe_{\mathfrak{c}_{\epsilon}}[\tau]$. To achieve better reliability, i.e., smaller $\epsilon$, in general the agent requires larger $\expe_{\mathfrak{c}_{\epsilon}}[\tau]$.

\begin{definition}[Adaptivity Gain]
The adaptivity gain is defined as the best reduction in the expected number of measurements when searching with an $\epsilon$-reliable adaptive strategy $\mathfrak{c}^{\prime}_{\epsilon} \in \mathcal{C}^{A}_{\epsilon}$, over an $\epsilon$-reliable non-adaptive strategy $\mathfrak{c}_{\epsilon} \in \mathcal{C}^{NA}_{\epsilon}$. Mathematically, it is given as
\begin{align}
\min_{\mathfrak{c}_{\epsilon} \in \mathcal{C}^{NA}_{\epsilon}}\expe[\tau] - \min_{\mathfrak{c}^{\prime}_{\epsilon} \in \mathcal{C}^{A}_{\epsilon}}\expe[\tau^{\prime}].
\end{align}
\end{definition}

Hence, characterizing adaptivity gain allows us to characterize the improvement in target acquisition rate when using adaptive strategies over non-adaptive strategies.

\section{Preliminaries: Channel Coding with State and Feedback}
\label{sec:prelim}
In this section, we review fundamentals of channel coding with state and feedback and relevant literature to connect these information theoretic concepts to the problem of searching under measurement dependent noise discussed in the previous section. The aim is to formulate an equivalent model of channel coding with state and feedback for comparison to (\ref{eq:noisysearch}).

\label{ReviewChannelCoding}
\begin{figure}[!htb]
     \centering
     \includegraphics[ width=0.7\textwidth]{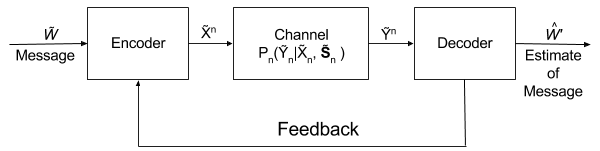}
     \caption{Transmission over a communication channel with state and feedback} 
 \label{fig:Basic}
 \end{figure}
A communication channel is specified by a set of inputs $\tilde{X} \in \tilde{\mathcal{X}}$, a set of outputs $\tilde{Y} \in \tilde{\mathcal{Y}}$, and a channel transition probability measure $\P(\tilde{y}|\tilde{x})$ for every $\tilde{x} \in \tilde{\mathcal{X}}$ and $\tilde{y} \in \tilde{\mathcal{Y}}$ that expresses the probability of observing a certain output $\tilde{y}$ given that an input $\tilde{x}$ was transmitted \cite{CoverBook2nd}. Throughout this work, we will concentrate on coding over a channel with state and feedback (section 4.6 in~\cite{Gallager}).  Formally, at time $n$ the channel state, $\tilde{\mbf{S}}_n$ belongs to a discrete and finite set $\tilde{\mathcal{A}}$. We assume that the channel state is known at both the encoder and the decoder. For a channel with state, the transition probability at time $n$ is specified by the conditional probability assignment $\P_n\left(\tilde{Y}_n |\tilde{X}_n, \tilde{\mbf{S}}_n \right)$. Transmission over such a channel is shown in Figure~\ref{fig:Basic}. In general, the channel state $\tilde{\mbf{S}}_n$ at time $n$ evolves as a function of all past outputs and all past states, 
\begin{equation} 
\label{eq:state}
\tilde{\mbf{S}}_n = \tilde{g}_n(\tilde{Y}_1, \tilde{Y}_2,\ldots, \tilde{Y}_{n-1}, \tilde{\mbf{S}}_1, \tilde{\mbf{S}}_2,\ldots, \tilde{\mbf{S}}_{n-1}).
\end{equation}
The goal is to encode and transmit a uniformly distributed message $\tilde{\mbf{W}} \in [M] $ over the channel. The encoding function $\phi_n$ at any time $n$ depends on the message to be transmitted $\tilde{\mbf{W}}$, all past states, and all the past outputs. Thus the next symbol to be transmitted is given by 
\begin{equation} 
\tilde{X}_n = \phi_n(\tilde{Y}_1, \tilde{Y}_2, \ldots, \tilde{Y}_{n-1}, \tilde{\mbf{S}}_1, \tilde{\mbf{S}}_2, ..., \tilde{\mbf{S}}_{n}, \tilde{\mbf{W}}).
\end{equation}
The encoder obtains the past outputs from the decoder due to the availability of a noiseless feedback channel from decoder to encoder. In this paper, we assume that both encoder and decoder know the evolution of the channel state, i.e., the sequence $\{\mbf{\mbf{S}}_n\}_{n \geq 1}$. After $\tau$  channel uses, the decoder uses the noisy observations $\tilde{\mbf{Y}}^{\tau}$ and state information $\{\tilde{\mbf{S}}_1, \tilde{\mbf{S}}_2, \ldots, \tilde{\mbf{S}}_{\tau}\}$ to find the best estimate $\tilde{\mbf{W}}^{\prime}$, of the message $\tilde{\mbf{W}}$. The probability of error  at the end of message transmission is given by $\Pe = \P(\tilde{\mbf{W}}^{\prime} \neq \tilde{\mbf{W}} | \tilde{\mbf{Y}}, \{\tilde{\mbf{S}}_1, \tilde{\mbf{S}}_2, \ldots, \tilde{\mbf{S}}_{\tau}\})$ and the average probability of error  is given by $\overline{\Pe} = \P(\tilde{\mbf{W}}^{\prime} \neq \tilde{\mbf{W}})$.

\begin{example}[Binary Additive White Gaussian Noise channel with State and feedback]
\label{ex:bawgn}
Consider a Binary Additive White Gaussian Noise (BAWGN) channel with noisy output $\tilde{Y}_n$ given as the sum of input $\tilde{X}_n \in \{0,1\}$ and  Gaussian random variable $\tilde{Z}_n \in \mathbb{R}$ whose distribution is a function of the channel state $\tilde{\mbf{S}}_n$. Specifically, $\tilde{Z}_n$ is a Gaussian random variable with state dependent noise variance $|\tilde{\mbf{S}}_n|\delta\sigma^2$ for some $\delta> 0$. 
In other words, we have
\begin{equation}
\tilde{Y}_n = \tilde{X}_{n} + \tilde{Z}_{n}(\tilde{\mbf{S}}_n), 
\label{eq:Gaussianoutput}
\end{equation}
where $\tilde{Z}_n \sim \mathcal{N}(0, |\tilde{\mbf{S}}_n|\delta\sigma^2)$, and the state evolves as $\tilde{\mbf{S}}_n = \tilde{g}_n(\tilde{Y}_1, \tilde{Y}_2, \ldots, \tilde{Y}_{n-1}, \tilde{\mbf{S}}_1, \tilde{\mbf{S}}_2, \ldots, \tilde{\mbf{S}}_{n-1})$. Transmission over a BAWGN channel is illustrated in Figure~\ref{fig:Gaussian}. 
\begin{figure}[!htb]
     \centering
     \includegraphics[ width=0.5\textwidth]{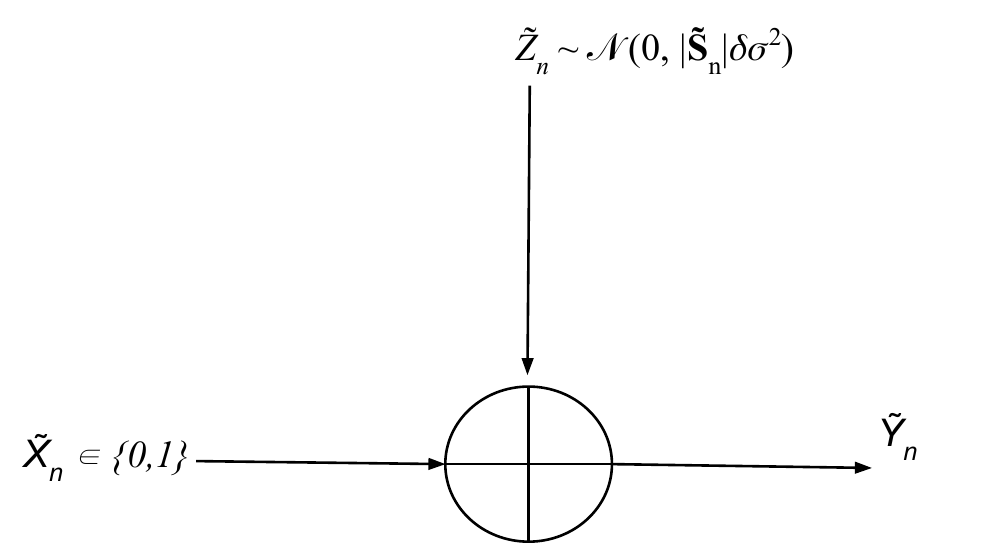}
     \caption{Transmission over a BAWGN channel with binary input $\tilde{X}_n$ and Gaussian noise $\tilde{Z}_n$.} 
 \label{fig:Gaussian}
 \end{figure}
\end{example}

\begin{proposition}
\label{prop:connection}
The problem of searching under measurement dependent Gaussian noise is equivalent to the problem of channel coding over a BAWGN channel with state and feedback. Specifically,  
\begin{itemize}
\item
The true location vector $\mbf{W}$ can be cast as a message $\tilde{\mbf{W}}$ to be transmitted over the BAWGN. Therefore, there are $\frac{B}{\delta}$ possible messages.

\item
An $\epsilon$-reliable search strategy $\mathfrak{c}_{\epsilon}$ provides a sequence of $\{g_1, g_2, \ldots, g_{\tau}\}$ such that $\P(\tilde{\mbf{W}}^{\prime} \neq \tilde{\mbf{W}}) \leq \epsilon$. Hence, setting $\tilde{g}_i = g_i$ for all $i \in \{1, 2, \ldots, \tau\}$, the search strategy dictates the evolution of channel states $\tilde{\mbf{S}}_n$.

\item
The measurement matrix $\overline{\mbf{S}}^{\tau}$ can be used as the codebook, i.e., by setting  $\{\tilde{\mbf{S}}_1, \tilde{\mbf{S}}_2, \ldots, \tilde{\mbf{S}}_{\tau}\} = \overline{\mbf{S}}$. Specifically, codewords are obtained by setting $\tilde{X}_n = \phi_n(\tilde{\mbf{Y}}^{n-1},\tilde{ \mbf{S}}_1, \tilde{\mbf{S}}_2, \ldots, \tilde{\mbf{S}}_n, \tilde{\mbf{W}}) = \tilde{\mbf{W}}^{\intercal}\tilde{\mbf{S}}_n$.

\item
The measurement vector fixes the channel transition probability measure as $\P(\tilde{Y}_n| \tilde{x}_n,\tilde{\textbf{S}}_n) = \mathcal{N}(\tilde{x}_n, |\tilde{\textbf{S}}_n|\delta\sigma^2)$ since noise distribution is  $\tilde{Z}_n \sim \mathcal{N}(0, |\tilde{\mbf{S}}_n|\delta \sigma^2)$ for $\tilde{x}_n \in \{ 0, 1\}$. Hence, the channel state depends on measurement vector.

\end{itemize}
\end{proposition}

A coding scheme for a channel with state and feedback can double as a search strategy. This general approach of search using channel codes provides an efficient way to design and
compare non-adaptive and adaptive search strategies. This also implies that feedback can improve the capacity of a channel with state which is what we characterize as our adaptivity gain for the problem of searching under measurement dependent noise.

\begin{definition}
The BAWGN capacity with input distribution $\ber(q)$ and noise variance $\sigma^2$ is defined as 
\begin{align}
C_{\text{BAWGN}}\left(q, \sigma^2\right) 
&:= -\int_{-\infty}^{\infty} \left( (1-q) G(y; 0, \sigma^2) + q G(y; 1, \sigma^2)\right) \times
\nonumber
\\
&\hspace{0.5cm}\times \log \left( (1-q) G(y; 0, \sigma^2) + q G(y; 1, \sigma^2)\right)
\nonumber
\\
& \hspace{0.5cm}- \frac{1}{2}\log(2\pi e  \sigma^2).
\end{align}
\end{definition}

\begin{corollary}
\label{cor:channelstate}
From channel coding over a BAWGN channel with state and feedback, we obtain that for any small $\xi > 0$ and $n$ large enough, there exists an $\epsilon$-reliable search strategy $\mathfrak{c}_{\epsilon}$ such the following holds
\begin{align}
\expe_{\mathfrak{c}_{\epsilon}}[\tau] &\leq n, \\
2^{n(C_{\text{BAWGN}}(\frac{1}{2}, \frac{B\sigma^2}{2})-\xi)} &\overset{(a)}\leq \frac{B}{\delta} \overset{(b)}< 2^{nC_{\text{BAWGN}}(\frac{1}{2}, \delta \sigma^2)},
\end{align}
where $(a)$ follows from Theorem 4.6.1 in~\cite{Gallager} and $(b)$ follows by combining the fact that the best channel is obtained when noise variance is the least, i.e., $\delta \sigma^2$, with the converse of the noisy channel coding theorem~\cite{CoverBook2nd}.
\end{corollary}

\color{black}
\section{Main Results}
\label{sec:main_results}

In this section, we characterize a lower bound on the adaptivity gain 
$\min_{\mathfrak{c}_{\epsilon} \in \mathcal{C}^{NA}_{\epsilon}}\expe[\tau] - \min_{\mathfrak{c}^{\prime}_{\epsilon} \in \mathcal{C}^{A}_{\epsilon}}\expe[\tau^{\prime}]$; the performance improvement measured in terms of reduction in the expected number of measurements for searching over a width $B$ among $\frac{B}{\delta}$ locations under measurement dependent Gaussian noise.

\begin{theorem}
\label{thm:gain_lower_bound}
Let $\epsilon \in (0,1)$. For any $\epsilon$-reliable non-adaptive strategy $\mathfrak{c}_{\epsilon} \in \mathcal{C}^{NA}_{\epsilon}$ searching over a search region of width $B$ among $\frac{B}{\delta}$ locations with $\tau$ number of measurements, there exists an $\epsilon$-reliable adaptive strategy $\mathfrak{c}^{\prime}_{\epsilon} \in \mathcal{C}^{A}_{\epsilon}$ with $\tau^{\prime}$ number of measurements, such that for some small constant $\eta > 0$ the following holds
\begin{align*}
\expe_{\mathfrak{c}_{\epsilon}}[\tau] - \expe_{\mathfrak{c}^{\prime}_{\epsilon}}[\tau^{\prime}]
&
\geq \max_{\alpha \in \mathcal{I}_{\frac{B}{\delta}}} \left\{ 
\log \frac{1}{\alpha} \left(\frac{(1-\epsilon) }{C_{\text{BAWGN}}(q^{\ast}, q^{\ast}B \sigma^2)}  \right.   -\frac{1}{C_{\text{BAWGN}} \left( q^{\ast}, q^{\ast} B \sigma^2\right) -  \eta}
\right)
\nonumber
\\
&
\quad + \log \frac{\alpha B}{\delta} \left(\frac{(1-\epsilon) }{
C_{\text{BAWGN}}(q^{\ast}, q^{\ast}B\sigma^2)} 
 - \frac{1}{
C_{\text{BAWGN}} \left(\frac{1}{2}, \frac{\alpha B \sigma^2}{2}\right)-\eta} \right)
\\
\nonumber
&\left. \quad
- h(B, \delta, \sigma^2, \alpha, \epsilon, \eta)  \right\},
\end{align*}
where
\begin{align*}
h(B, \delta, \sigma^2, \alpha, \epsilon, \eta)
&=
\frac{\log \left( \frac{2}{\epsilon}\right) + \log \log \left( \frac{1}{\alpha}\right) + a_{\eta}}{C_{\text{BAWGN}} \left( q^{\ast}, q^{\ast}B \sigma^2\right) -\eta} 
+ \frac{\log \left( \frac{2}{\epsilon}\right) +  \log \log \left( \frac{\alpha B}{\delta }\right) + a_{\eta}}{C_{\text{BAWGN}}\left(\frac{1}{2}, \frac{\alpha B \sigma^2}{2} \right) - \eta}
\nonumber
\\
& \quad + \frac{h(\epsilon)}{C_{\text{BAWGN}}^{(B, \delta, \sigma^2)}(q^{\ast}, q^{\ast}B \sigma^2)} ,
\end{align*} 
$
q^{*}
=
\argmax_{q \in \mathcal{I}_{\frac{B}{\delta}}} C_{\text{BAWGN}}(q, qB\sigma^2),
$
and $a_{\eta}$ is the solution of the following equation
\begin{align}
\eta =\frac{a}{a-3}\max_{q \in \mathcal{I}_{\frac{B}{\delta}}}\int_{-\infty}^{\infty} \frac{e^{-\frac{y^2}{2Bq\sigma^2}}}{\sqrt{2 \pi qB \sigma^2}} \left[ \frac{2y-1}{2qB\sigma^2}\right]_{(a-3)} dy .
\end{align}

\end{theorem}

Proof of Theorem~\ref{thm:gain_lower_bound} is obtained by combining Lemma~\ref{lemma:converse_k_1} and  Lemma~\ref{lemma:achv}. Theorem~\ref{thm:gain_lower_bound} provides a non-asymptotic lower bound on adaptivity gain. The bound can be viewed as two parts corresponding to two stages. Intuitively, the first part corresponds to the initial stage of the search, where the agent narrows down the target's location to some coarse $\alpha$ fractions of the total search region, i.e., narrows to a section of width $\alpha B$ with high confidence. The second stage corresponds to refined the search within one of the coarse sections $\alpha B$ obtained from initial stage. This implies that an adaptive strategy can zoom in and confine the search to a smaller section to reduce the noise intensity. Whereas, a non adaptive strategy does not adapt to zoom in, and thus performs equally in both stages. We formalize this intuition in Lemma~\ref{lemma:achv}. Optimizing over $\alpha$ fraction of the first search we obtain a bound on expected number of measurements.  We obtain the following corollary as a consequence of Theorem~\ref{thm:gain_lower_bound}.

\begin{corollary}
\label{cor:two_regime_gains}
Let $\epsilon \in (0,1)$. For any $\epsilon$-reliable non-adaptive strategy $\mathfrak{c}_{\epsilon} \in \mathcal{C}^{NA}_{\epsilon}$ searching over a search region of width $B$ among $\frac{B}{\delta}$ with $\tau$ number of measurements, there exists an $\epsilon$-reliable adaptive strategy $\mathfrak{c}^{\prime}_{\epsilon} \in \mathcal{C}^{A}_{\epsilon}$ with $\tau^{\prime}$ number of measurements, such that for a fixed $B$ the asymptotic adaptivity gain grows logarithmically with the total number of locations,
\begin{align}
\label{eq:delta_gain}
\lim_{\delta \to 0} \frac{\expe_{\mathfrak{c}_{\epsilon}}[\tau] - \expe_{\mathfrak{c}^{\prime}_{\epsilon}}[\tau^{\prime}]}{\log \frac{B}{\delta}} 
\geq 
\frac{1-\epsilon}{ C_{\text{BAWGN}}(q^{\ast}, q^{\ast}B \sigma^2)} - 1.
\end{align} 
For a fixed $\delta$, the asymptotic adaptivity gain grows at least linearly with total number of locations, 
\begin{align}
\label{eq:B_gain}
\lim_{B \to \infty} \frac{\expe_{\mathfrak{c}_{\epsilon}}[\tau] - \expe_{\mathfrak{c}^{\prime}_{\epsilon}}[\tau^{\prime}]}{\frac{B}{\delta} \log \frac{B}{\delta}} 
\geq 
\frac{(1-\epsilon) \sigma^2 \delta }{\log e}.
\end{align}
Furthermore, we have
\begin{align}
\label{eq:B_NA}
\lim_{B \to \infty}  \frac{\min_{\mathfrak{c}_{\epsilon} \in \mathcal{C}_{\epsilon}^{NA}}\expe_{\mathfrak{c}_{\epsilon}}[\tau]}{\frac{B}{\delta} \log \frac{B}{\delta}} 
\geq 
\frac{(1-\epsilon) \sigma^2 \delta }{\log e},
\end{align} 
and
\begin{align}
\label{eq:B_A}
\lim_{B \to \infty} \frac{\min_{\mathfrak{c}_{\epsilon} \in \mathcal{C}_{\epsilon}^{A}} \expe_{\mathfrak{c}^{\prime}_{\epsilon}}[\tau^{\prime}]}{\frac{B}{\delta} } 
= 0.
\end{align} 

\end{corollary}
The proof of the above corollary is provided in~Appendix-C.

\begin{remarks}
The above corollary characterizes the two qualitatively different regimes previously discussed. For fixed $B$, as $\delta$ goes to zero the asymptotic adaptivity gain scales as only $\log \frac{B}{\delta}$, whereas for fixed $\delta$, as $B$ increases the asymptotic adaptivity gain scales as $\frac{B}{\delta} \log \frac{B}{\delta}$. In other words, target acquisition rate improves by a constant for fixed $B$ as $\delta$ decreases while it grows linearly with $B$ for a fixed $\delta$. In other words, adaptivity provides a larger gain in target acquisition rate for the regime where the total search width is growing than in the case where we fix the total width and shrink the location widths. In Section~\ref{sec:num_results} we related this phenomenon to the diminishing capacity of BAWGN channel when the total noise $\frac{B\sigma^2}{2}$ grows.
\end{remarks}

Next we provide the main technical components of the proof of Theorem~\ref{thm:gain_lower_bound}.

\subsection{Converse: Non-Adaptive Search Strategies}

\begin{lemma}
\label{lemma:converse_k_1}

The minimum expected number of measurements required for any $\epsilon$-reliable non-adaptive search strategy can be lower bounded as
\begin{align*}
\min_{\mathfrak{c}_{\epsilon} \in \mathcal{C}^{NA}_{\epsilon}}\expe_{\mathfrak{c}_{\epsilon}}[\tau] \geq \frac{(1-\epsilon) \log \left(\frac{B}{\delta}\right) -h(\epsilon)}{C_{\text{BAWGN}}\left(q^{\ast}, q^{\ast} B \sigma^2 \right)}.
\end{align*}
\end{lemma}

Proof of the Lemma~\ref{lemma:converse_k_1} is provided in~Appendix-A. The proof follows from the fact that clean signal $X_i$ and noise $Z_i$ are independent over time and independent of  past observations for $i = 1,2, \ldots, n$,  due to the non-adaptive nature of the search strategy. In the absence of information from past observation outcomes, the agent tries to maximize the mutual information $I(X_i, Y_i)$ at every measurement. Since $X_i \sim \ber(q_i)$ and $Z_i \sim \mathcal{N}(0, q_i  B\sigma^2)$, the mutual information $I(X_i, Y_i) = C_{\text{BAWGN}}\left(q_i, q_i B \sigma^2 \right)$ is maximized at $q_i = q^{\ast}$.

\subsection{Achievability: Adaptive Search Strategy}

Consider the following two stage search strategy.

\subsubsection{First Stage (Fixed Composition Strategy $\mathfrak{c}^{1}_{\frac{\epsilon}{2}}$)} We group the $\frac{B}{\delta}$ locations of width $\delta$ into $\frac{1}{\alpha}$ sections of width $\alpha B$. Let $\mbf{W}^{\prime}$ denote the true location of the target among the sections of width $ \alpha B $. Now, we use a non-adaptive strategy to search for the target location among $\frac{1}{\alpha}$ sections of width $\alpha B$. In particular, we use a fixed composition strategy where at every time instant $n$, the fraction of total locations probed is fixed to be $q^{\ast}$. In other words, the measurement vector $\mbf{S}^{\prime}_n$ at every instant $n$ is picked uniformly randomly from the set of measurement vectors $\{\mbf{S}^{\prime} \in \mathcal{U}_{\frac{1}{\alpha}}: |\mbf{S}^{\prime}| = \lfloor \frac{q^{\ast}}{\alpha} \rfloor \}$. For the ease of exposition, we assume that $\frac{q^{\ast}}{\alpha}$ is an integer.  Hence, for this strategy, at every $n$, $X_n \sim \ber(q^{\ast})$ and $Z_n \sim \mathcal{N}(0, q^{\ast}B\sigma^2)$. For all $i \in \{1,2, \ldots, \frac{1}{\alpha} \}$, let $\mbs{\rho}^{\prime}_n(i)$ be the posterior probability of the estimate $\hat{\mbf{W}}^{\prime}(i) = 1$ after reception of $\mbf{Y}^{n-1}$, i.e.,  $\mbs{\rho}^{\prime}_n(i): = \P \left( \hat{\mbf{W}}^{\prime}(i) = 1| \mbs{Y}^{n-1} \right)$ and let $\mbs{\rho}^{\prime}_n: = \left\{\mbs{\rho}^{\prime}_n(1), \mbs{\rho}^{\prime}_n(2), \ldots, \mbs{\rho}^{\prime}_{n}\left(\frac{1}{\alpha}\right) \right\}$. Assume that agent begins with a uniform probability over the $\frac{1}{\alpha}$ sections, i.e., $\mbs{\rho}^{\prime}_0 = \{\alpha, \alpha, \ldots, \alpha \}$. The  posterior probability $\mbs{\rho}^{\prime}_{n+1}(i)$ at time $n+1$ when $Y_n = y$ is obtained by the following Bayesian update:
\begin{align}
\label{eq:rho_update1}
\mbs{\rho}^{\prime}_{n+1}(i)
= 
\left\{
\begin{array}{ll}
\frac{\mbs{\rho}^{\prime}_{n}(i) G(y; 1, q^{\ast}B\sigma^2)}{\mathcal{D}^{\prime}_n} & \text{if } \mbs{S}^{\prime}_n(i) = 1,\\
\frac{\mbs{\rho}^{\prime}_{n}(i) G(y; 0, q^{\ast}B\sigma^2)}{\mathcal{D}^{\prime}_n} & \text{if } \mbs{S}^{\prime}_n(i) = 0,
\end{array}
\right.
\end{align}
where
\begin{align}
\label{eq:rho_norm1}
\mathcal{D}^{\prime}_n 
= \sum_{j: \mbf{1}_{\{\mbs{S}_n(j) = 1\}}}\mbs{\rho}^{\prime}_{n}(j) G(y; 1, q^{\ast}B\sigma^2)
+ \sum_{j: \mbf{1}_{\{\mbs{S}_n(j) = 0\}}}\mbs{\rho}^{\prime}_{n}(j) G(y; 0, q^{\ast}B\sigma^2).
\end{align}

Let $\tau^{1} : = \inf\left\{n: \max_{i} \mbs{\rho}^{\prime}_n(i) \geq 1- \frac{\epsilon}{2} \right\}$ be the number of measurements used under stage 1. Note that $\tau^{1}$ is a random variable. Hence, first stage is a non-adaptive variable length strategy. Now, the expected stopping time $\expe_{\mathfrak{c}^{1}_{\frac{\epsilon}{2}}}[\tau^{1}]$ can be upper bounded using Lemma~\ref{lemm:stage_1_time} from Appendix-B. 

\subsubsection{Second Stage (Sorted Posterior Matching Strategy $\mathfrak{c}_{\frac{\epsilon}{2}}^2$)} In the second stage, the agent zooms into the $\alpha B$ width section obtained from the first stage and uses an adaptive strategy to search only within this $\alpha B$ section. The agent searches for the target location of width $\delta$ among the remaining $\frac{\alpha B}{\delta}$ locations. In particular, we use the sorted posterior matching strategy proposed in~\cite{SungEnChiu} which we describe next. Let $\mbf{W}^{\prime \prime}$ denote the true target location of width $\delta$. For all $i \in \{1,2, \ldots, \frac{\alpha B}{\delta} \}$, let $\mbs{\rho}^{\prime \prime}_n(i)$ be the posterior probability of the estimate $\hat{\mbf{W}}^{\prime \prime}(i) = 1$ after reception of $\mbf{Y}^{n-1}$, i.e., $\mbs{\rho}^{\prime}_n(i): = \P \left( \hat{\mbf{W}}^{\prime \prime}(i) = 1| \mbf{Y}^{n-1} \right)$ and let $\mbs{\rho}^{\prime \prime}(n): = \{\mbs{\rho}^{\prime \prime}_n(1), \mbs{\rho}^{\prime \prime}_n(2), \ldots, \mbs{\rho}^{\prime \prime}_{n}\left(\frac{\alpha B}{\delta}\right) \}$. Assume the agent begins with a uniform probability over the $\frac{\alpha B}{\delta}$ sections, i.e., $\mbs{\rho}^{\prime \prime}_0 = \left\{\frac{\delta}{\alpha B}, \frac{\delta}{\alpha B}, \ldots, \frac{\delta}{\alpha B} \right\}$. At every time instant $n$, we sort the posterior values in descending order to obtain the sorted posterior vector $\mbs{\rho}^{\downarrow}_n$. Let vector $I_n$ denote the corresponding ordering of the location indices in the new sorted posterior. Define
\begin{align}
k^{\ast}_n:=  \argmin_{i} \left| \sum_{j = 1}^{i} \mbs{\rho}^{\downarrow}_n(j)- \frac{1}{2}\right|.
 \end{align}
We choose the measurement vector $\mbf{S}_n^{\prime \prime}$ such that $\mbf{S}_n^{\prime \prime}(j) = 1$ if and only if $j \in \{I_n(1),I_n(2), \ldots, I_n(k^{\ast}_n)\}$. Note that for this strategy, at every $n$, the noise is $Z_n \sim \mathcal{N}(0, |\mbf{S}^{\prime \prime}_n|\delta \sigma^2)$ and the worst noise intensity is $\mathcal{N}(0, \frac{\alpha B \sigma^2}{2})$. The  posterior probability $\mbs{\rho}^{\prime \prime}_{n+1}(i)$ at time $n+1$ when $Y_n = y$ is obtained by the following Bayesian update:
\begin{align}
\label{eq:rho_update2}
\mbs{\rho}^{\prime \prime}_{n+1}(i)
= 
\left\{
\begin{array}{ll}
\frac{\mbs{\rho}^{\prime \prime}_{n}(i) G(y; 1, |\mbf{S}^{\prime \prime}_n|\delta \sigma^2)}{\mathcal{D}^{\prime \prime}_n} & \text{if } \mbs{S}^{\prime \prime}_n(i) = 1,\\
\frac{\mbs{\rho}^{\prime \prime}_{n}(i) G(y; 0, |\mbf{S}^{\prime \prime}_n|\delta \sigma^2)}{\mathcal{D}^{\prime \prime}_n} & \text{if } \mbs{S}^{\prime \prime}_n(i) = 0,
\end{array}
\right.
\end{align}
where
\begin{align}
\label{eq:rho_norm2}
\mathcal{D}^{\prime \prime}_n 
= \sum_{j: \mbf{1}_{\{\mbs{S}_n(j) = 1\}}}\mbs{\rho}^{\prime \prime}_{n}(j) G\left(y; 1,|\mbf{S}^{\prime \prime}_n|\delta \sigma^2\right)
+ \sum_{j: \mbf{1}_{\{\mbs{S}_n(j) = 0\}}}\mbs{\rho}^{\prime \prime}_{n}(j) G\left(y; 0, |\mbf{S}^{\prime \prime}_n|\delta \sigma^2\right).
\end{align}
Let $\tau^{2} : = \inf\left\{n: \max_{i} \mbs{\rho}^2_n(i) \geq 1- \frac{\epsilon}{2} \right \}$ be the number of measurements used under stage 2. Note that $\tau^{2}$ is a random variable. Hence, the second stage is an adaptive variable length strategy. The expected number of measurements $\expe_{\mathfrak{c}^{2}_{\frac{\epsilon}{2}}}[\tau^{\prime \prime}]$ can be upper bounded using Lemma~\ref{lemm:sortPM_tau} from Appendix-B.

Noting that the total probability of error of the two stage search strategy is less than $\epsilon$ and that the expected stopping time is $\expe_{\mathfrak{c}^{\prime}_{\epsilon}}[\tau^{\prime}] = \expe_{\mathfrak{c}^1_{\frac{\epsilon}{2}}}[\tau^{1}]+ \expe_{\mathfrak{c}^2_{\frac{\epsilon}{2}}}[\tau^{2}]$, we have the assertion of the following lemma.

\begin{lemma}
\label{lemma:achv}
The minimum expected number of measurements required for the above $\epsilon$-reliable adaptive search strategy $\mathfrak{c}^{\prime}_{\epsilon}$ can be upper bounded as
\begin{align}
\expe_{\mathfrak{c}^{\prime}_{\epsilon}}[\tau^{\prime}]
\leq 
\min_{\alpha \in \mathcal{I}_{\frac{B}{\delta}}}\left\{\frac{\log \frac{1}{\alpha} + \log \frac{2}{\epsilon} + \log \log \frac{1}{\alpha} + a_{\eta}}{C_{\text{BAWGN}}\left(q^{\ast}, q^{\ast} B \sigma^2 \right) - 
\eta}
+
\frac{\log \frac{\alpha B}{\delta} + \log \frac{2}{\epsilon} + \log \log \frac{\alpha B}{\delta} + a_{\eta}}{C_{\text{BAWGN}}\left(\frac{1}{2}, \frac{\alpha B \sigma^2}{2} \right) - 
\eta} \right\}.
\end{align}
\end{lemma}

\begin{remarks}
For an $\epsilon$-reliable adaptive search strategy $\mathfrak{c}^{\prime}_{\epsilon} \in \mathcal{C}^{A}_{\epsilon}$ using the two stage strategy, the non-asymptotic upper bound provided by Lemma~\ref{lemma:achv} for $\min_{\mathfrak{c}^{\prime}_{\epsilon} \in \mathcal{C}_{\epsilon}^{A}} \expe^{\prime}_{\mathfrak{c}_{\epsilon}}[\tau^{\prime}]$  is tighter than the upper bound provided in~\cite{SungEnChiu} using the sorted posterior matching strategy. In fact, for any given $\alpha$, our bound is significantly smaller than the upper bound in~\cite{SungEnChiu}. In the asymptotically dominating terms of the order $\log \frac{B}{\delta}$, our upper bound closely follows the simulations as illustrated in Section~\ref{sec:num_results}.
\end{remarks}

\begin{remarks}
In the regime of fixed $B$ and diminishing $\delta$, Lemma~\ref{lemma:achv} together with Corollary~\ref{cor:channelstate} establishes the optimality of our proposed algorithm. Further, it characterizes a lower bound on the increase in targeting capacity when utilizing an adaptive strategy over the non-adaptive strategies.
\end{remarks}

\section{Extensions and Generalizations}

\subsection{Generalization to other noise models}
The main results presented in this paper consider the setup where the noise $Z_n$ is distributed as $\mathcal{N}(0, |\mbf{S}_n|\delta \sigma^2)$. In other words, the variance of the noise given by $(|\textbf{S}_n|\delta \sigma^2)$ is a linear function of the size of a measurement vector $|\mbf{S}_n|$. This model assumption holds when each target location adds noise equally and independently of other locations when probed together. In general, due to correlation across locations the additive noise variance can be assumed to scale as a non-decreasing function $f(\cdot)$ of the measurement vector $|\mbf{S}_n|$. In this section, we extend our model to a general formulation for the noise $Z_n \sim \mathcal{N}(0, f(|\mbf{S}_n|)\delta \sigma^2)$, where $f(\cdot)$ is a non-decreasing function of $|\mbf{S}_n|$. For example, $f(\textbf{S}_n) = |\mbf{S}_n|^{\gamma}$ for some $\gamma > 0$. Figure~\ref{fig:capacity}, shows that the effect of the noise function $f(|\mbf{S}_n|)$ on the capacity. 
 \begin{figure}[!htb]
     \centering
     \includegraphics[width=0.7\textwidth]{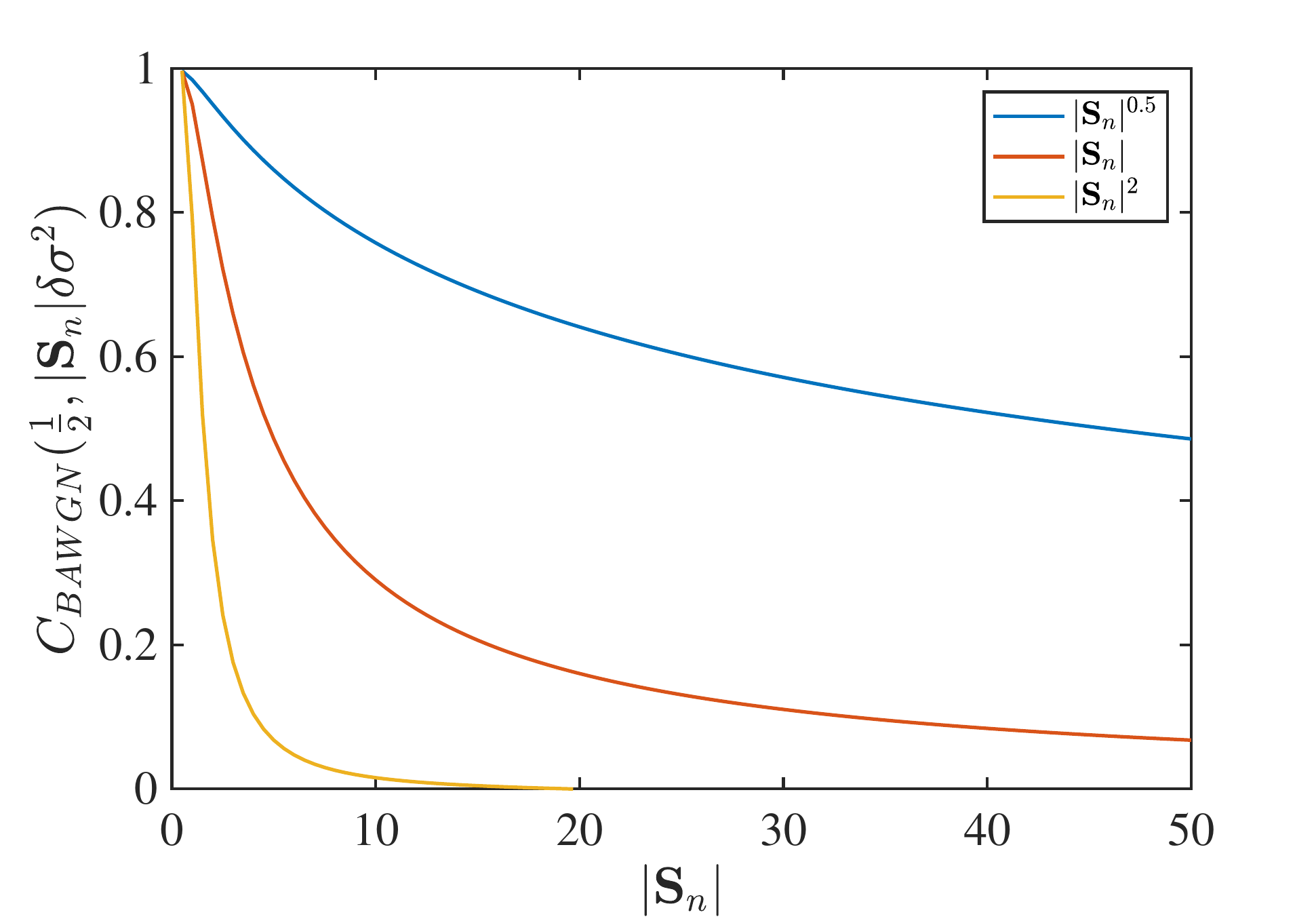}
     \caption{Behavior of capacity of BAWGN channel with $\sigma^2=0.25$ over a total search region of width $B=10$, location width $\delta =0.1$, as a function of the size of a measurement $|S_n|$.} 
       \label{fig:capacity}
 \end{figure}

\begin{theorem}
\label{thm:fgain_lower_bound}
Let $\epsilon \in (0,1)$ and let $f(\cdot)$ be a non-decreasing function. For any $\epsilon$-reliable non-adaptive strategy $\mathfrak{c}_{\epsilon} \in \mathcal{C}^{NA}_{\epsilon}$ searching over a search region of width $B$ among $\frac{B}{\delta}$ locations with $\tau$ number of measurements, there exists an $\epsilon$-reliable adaptive strategy $\mathfrak{c}^{\prime}_{\epsilon} \in \mathcal{C}^{A}_{\epsilon}$ with $\tau^{\prime}$ number of measurements, such that for some small constant $\eta > 0$ the following holds
\begin{align*}
\expe_{\mathfrak{c}_{\epsilon}}[\tau] - \expe_{\mathfrak{c}^{\prime}_{\epsilon}}[\tau^{\prime}]
&
\geq \left(\max_{\alpha \in \mathcal{I}_{\frac{B}{\delta}}} \left\{ 
\log \frac{1}{\alpha} \left(\frac{(1-\epsilon) }{C_{\text{BAWGN}}(q^{\ast}, f(\frac{q^{\ast} B}{\delta})\delta \sigma^2)}  \right. \right.  -\frac{1}{C_{\text{BAWGN}} \left( q^{\ast}, f(\frac{q^{\ast} B}{\delta})\delta \sigma^2\right) -  \eta}
\right)
\nonumber
\\
&
\quad + \log \frac{\alpha B}{\delta} \left(\frac{(1-\epsilon) }{
C_{\text{BAWGN}}(q^{\ast}, f(\frac{q^{\ast} B}{\delta})\delta\sigma^2)} \right.
\nonumber
\\
& \left. \left. \left. \hspace{1 cm}  - \frac{1}{
C_{\text{BAWGN}} \left(\frac{1}{2}, f(\frac{\alpha B}{2\delta})\delta \sigma^2\right)-\eta} \right) \right\} \right) (1+o(1)),
\end{align*}
where
$
q^{*}
=
\argmax_{q \in \mathcal{I}_{\frac{B}{\delta}}} C_{\text{BAWGN}}(q, f(\frac{qB}{\delta})\delta\sigma^2),
$
and $o(1)$ goes to 0 as $\frac{B}{\delta} \to \infty$.
\end{theorem}

\subsection{Multiple Targets}
\label{generalsetup}

The problem formulation and the main results of this paper consider the special case when there exists a single stationary target. Suppose instead the agent aims to find the true location of $r$ unique targets quickly and reliably. Our problem formulation is easily extended to the general case where there may exist multiple targets. In our generalization to multiple targets under the linear noise model (\ref{eq:noisysearch}), the clean signal indicates the the number of targets present in the measurement vector $\textbf{S}_n$. In particular, let $\mbf{W}^{(i)} \in \mathcal{U}_{\frac{B}{\delta}}$ be such that $\mbf{W}^{(i)}(j) = 1$ if and only if $j$-th location contains the $i$-th target. Then, the noisy observation is given as
\begin{align}
Y_n = \sum_{i = 1}^{r} (\mbf{W}^{(i)})^{\intercal}\mbf{S}_n + Z_n,
\end{align}
where $Z_n \sim \mathcal{N}(0, |\mbf{S}_n|\delta \sigma^2)$. Setting $X_n^{(i)} = (\mbf{W}^{(i)})^{\intercal}\mbf{S}_n $ for $i \in [r]$, we have
\begin{align}
Y_n = \sum_{i = 1}^{r} X^{(i)}_n + Z_n.
\end{align}
The problem of searching for multiple targets is equivalent to the problem of channel coding over a Multiple Access Channel (MAC) with state and feedback~\cite{nancy_asilomar}. In other words, we can extend the Proposition~1, to channel coding over a MAC with state and feedback with the following constraints: (i) $\mbf{W}^{(i)}$ can be viewed as the message to be transmitted by the $i$-th transmitter, (ii) the measurement matrix $\overline{\mbf{S}}_n$ can be viewed as the common codebook shared by all the transmitters, and (iii) a search strategy dictates the evolution of the MAC state. The channel transition is then fixed by the channel state which is measurement dependent.

\textbf{Example $\mathbf{1^{\prime}}$} (Establishing initial access in mm-Wave communications).
In the deployment of mm-Wave links into a  cellular or 802.11 network, the base station needs to to quickly switch between users and accommodate multiple mobile clients.
In this setup at time $n$ the noisy observation, $Y_n$, is a function of multiple users in the network, in addition to a measurement dependent noise.

\textbf{Example $\mbs{2^{\prime}}$} (Spectrum Sensing for Cognitive Radio). Consider the problem of opportunistically searching for $r$ vacant subbands of bandwidth $\delta$ over a total bandwidth of $B$. In this problem we desire to locate $r$ stationary vacant subbands quickly and reliably, by making measurements over time. Here again the noise intensity depends on the number of subbands probed, $\mbf{S}_n$, at each time instant $n$.

Searching for multiple targets with measurement dependent noise is a significantly harder problem compared to a single target case and achievability strategies for this problem even in the absence of noise are far more complex~\cite{Bshouty_kTargetsweighing, Chang_kTargetsCode} .

\section{Numerical Results}

\label{sec:num_results}
In this section we provide numerical analysis.


\subsection{Comparing Search Strategies}
In this section, we numerically compare four strategies proposed in the literature. Besides the sort PM strategy $\mathfrak{c}_{\epsilon}^2$ and the optimal variable length non-adaptive strategy i.e., the fixed composition strategy $\mathfrak{c}_{\epsilon}^1$,  we also consider two noisy variants of the binary search strategy. The noisy binary search applied to our search proceeds by selecting $\mbf{S}_n$ as half the width of the previous search region $\mbf{S}_{n-1}$ with higher posterior probability. The first variant we consider is fixed length noisy binary search, resembles the adaptive iterative hierarchical search strategy~\cite{7460513}, where each measurement vector $\mbf{S}_n$ is used $\alpha_{\epsilon}(\mbf{S}_n)|\mbf{S}_n|$ times where $\alpha_{\epsilon}(\mbf{S}_n)$ is chosen such that entire search result in an $\epsilon$-reliable search strategy. The second variant is variable length noisy binary search where each measurement vector $\mbf{S}_n$ is used until in each search we obtain error probability less than $\epsilon_p:=\frac{\epsilon}{\log{B/\delta}}$. Table~I provides a quick summary of the search strategies.

\begin{table}[!htb]
\centering
\caption{Candidate Search Strategies}
\label{Table:Strategies}
\begin{tabular}{|l|l|}
\hline
 & \\
Strategies $\mathfrak{c}_{\epsilon} \in \mathcal{C}_{\epsilon}$& Description of $\mbf{S}_n$ selection \\
\hline
Variable Length Random & $\bullet$ Select $\mbf{S}_n$ s.t. $|\mbf{S}_n|  = \frac{q^{\ast}B}{\delta}$ \T \\
& as dictated by strategy $\mathfrak{c}_{\epsilon}^1$\\
\hline
Fixed Length Noisy Binary& $\bullet$ Select $\mbf{S}_n$ as dictated by \T\\
& binary search strategy \\
 &$\bullet$ Repeat $\alpha_{\epsilon}(\mbf{S}_n)|\mbf{S}_n|$ times\\
\hline
Variable Length Noisy Binary& $\bullet$ Select $\mbf{S}_n$ as dictated by \T\\
& binary search strategy \\
&$\bullet$ Repeat $\tau$ times s.t.\\ 
&\setlength{\thickmuskip}{0mu} $\tau  = \min \{n: \|\mbs{\rho}_n \|{_{\scalebox{0.5}{$\infty$}}} \geq 1- \epsilon_p\}$\\
\hline
Sorted Posterior Matching&$\bullet$ Select $\mbf{S}_n$ as dictated by \T\\
& Sort PM strategy $\mathfrak{c}_{\epsilon}^2$\\
\hline
\end{tabular}
\end{table}

Figure~\ref{fig:EN_strategies}, shows the performance of each $\epsilon$-reliable search strategy, when considering fixed parameters $B$, $\delta$, and $\epsilon$. We note that the fixed length noisy binary strategy performs poorly in comparison to the optimal non-adaptive strategy. This shows that randomized non-adaptive search strategies such as the one considered in~\cite{Abari_AgileLink} perform better than both exhaustive search and iterative hierarchical search strategy. In particular, it performs better than variable length noisy binary search since when SNR is high since each measurement is repeated far too many times in order to be $\epsilon$-reliable. The performance of the optimal fully adaptive variable length strategies sort PM~\cite{SungEnChiu} is superior to all strategies even in the non-asymptotic regime.

\begin{figure}[ht]
     \centering
     \includegraphics[width=0.7\textwidth]{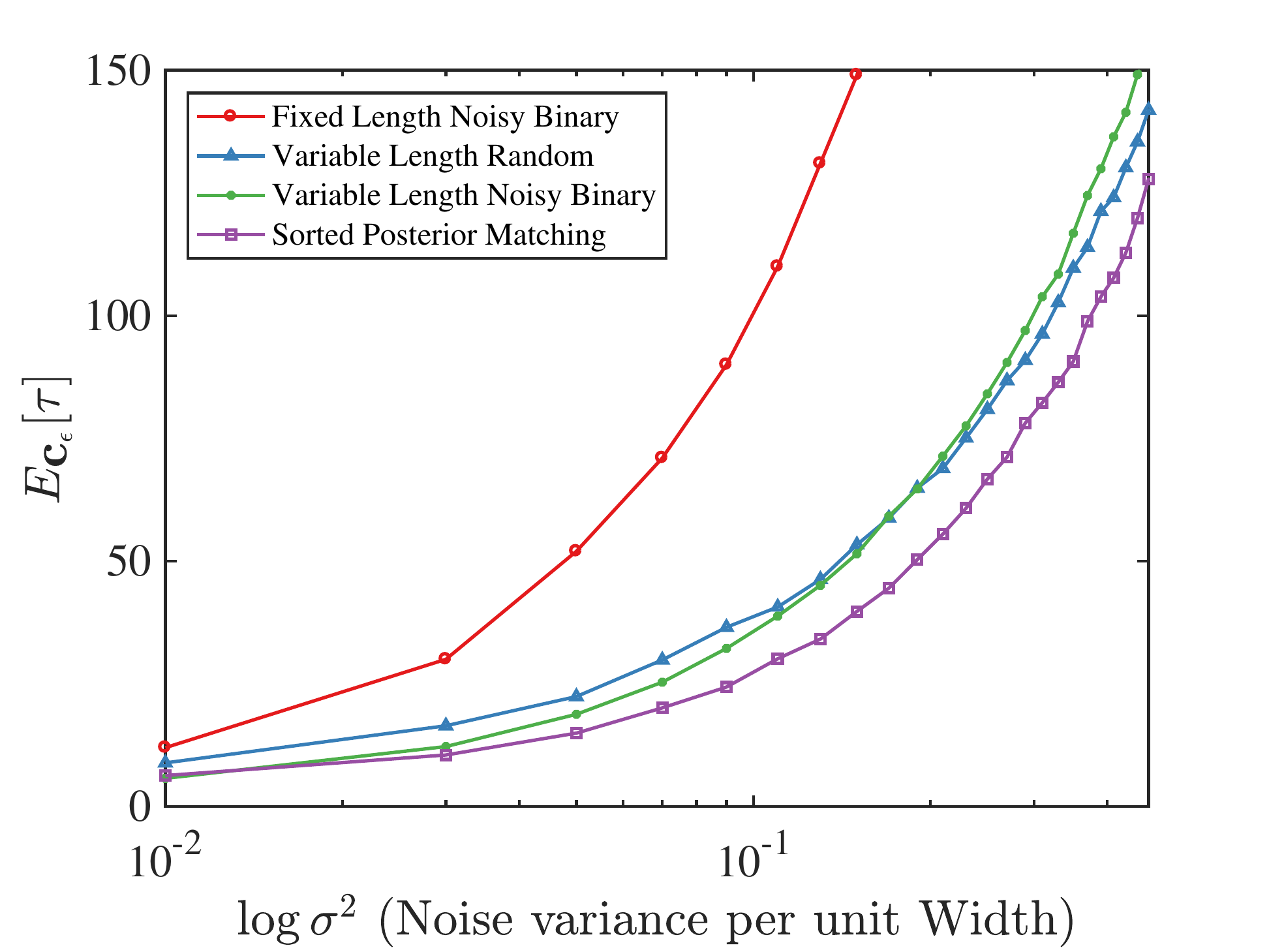}
     \caption{$\mathbb{E}_{\mathfrak{c}_{\epsilon}}[\tau]$ with $\epsilon = 10^{-4}$, $B=16$, and $\delta=1$, as a function of $\sigma^2$ for various strategies.} 
 \label{fig:EN_strategies}
 \end{figure}

\subsection{Two Distinct Regimes of Operation}
\label{sect:mainsimulations}
In this section, for a fixed $\sigma^2$ we are interested in the expected number of measurements required $\mathbb{E}_{\mathfrak{c}_{\epsilon}}[\tau]$ by an $\epsilon$-reliable strategy $\mathfrak{c}_{\epsilon}$, in the following two regimes: varying $\delta$ while keeping $B$ fixed, and varying $B$ while keeping $\delta$ fixed. Figures~\ref{fig:EN_varyB} and~\ref{fig:EN_varyDelta} show the simulation results of $\mathbb{E}_{\mathfrak{c}_{\epsilon}}[\tau]$ as a function of width $B$ and resolution $\delta$ respectively, for the fixed composition non adaptive strategy $\mathfrak{c}_{\epsilon}\in \mathcal{C}_{\epsilon}^{NA}$ and for the sort PM adaptive strategy $\mathfrak{c}_{\epsilon}\in \mathcal{C}_{\epsilon}^{A}$, along with dominant terms of the lower bound of Lemma~\ref{lemma:converse_k_1}, and the upper bound of Lemma~\ref{lemma:achv} for a fixed noise per unit width $\sigma^2=0.25$. For both of these cases, we see that the adaptivity gain grows as the total number of locations increases; however in distinctly different manner as seen in Corollary~\ref{cor:two_regime_gains}.

\begin{figure}[ht]
     \centering
     \includegraphics[width=0.7\textwidth]{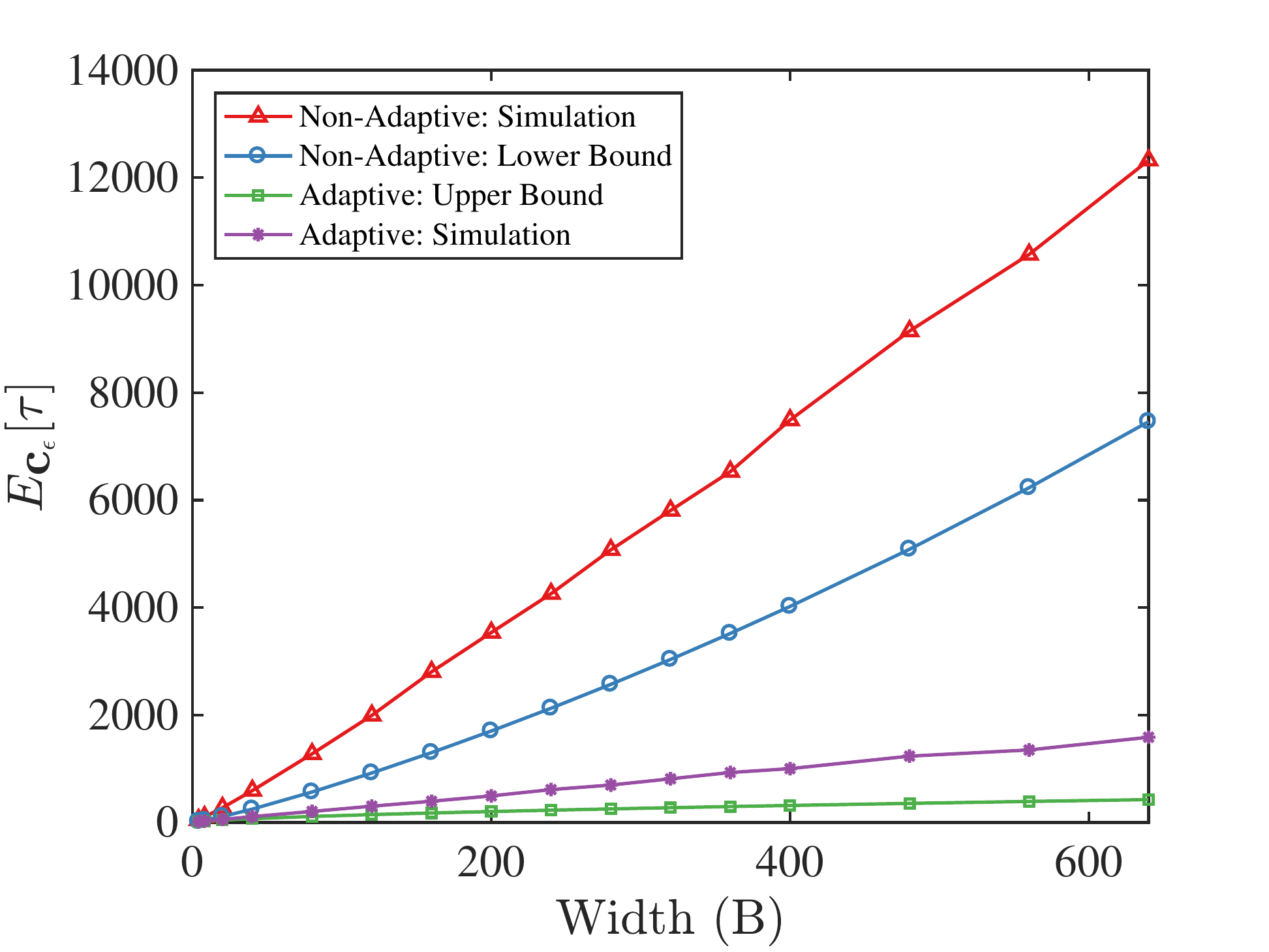}
     \caption{$\mathbb{E}_{\mathfrak{c}_{\epsilon}}[\tau]$ with $\epsilon = 10^{-4}$, $\sigma^2=0.25$, and $\delta=1$, as a function of B.} 
 \label{fig:EN_varyB}
 \end{figure}

 \begin{figure}[ht]
     \centering
     \includegraphics[width=0.7\textwidth]{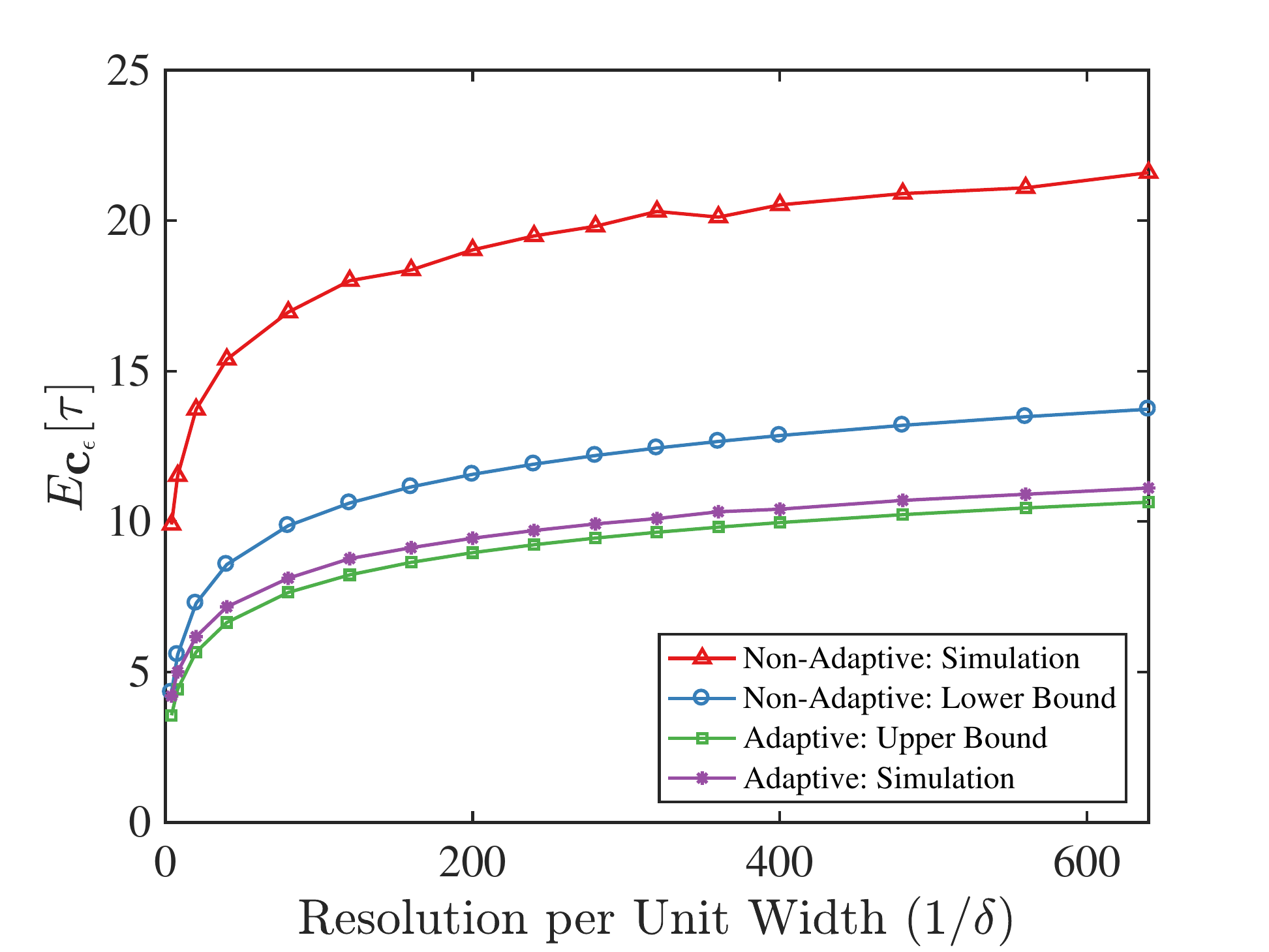}
     \caption{$\mathbb{E}_{\mathfrak{c}_{\epsilon}}[\tau]$ with $\epsilon = 10^{-4}$, $\sigma^2=0.25$ and $B=1$, as a function of $\delta$.} 
 \label{fig:EN_varyDelta}
 \end{figure}

\subsection{Relating the Regimes of Operation to Capacity}
In this section, we attempt to relate these two regimes of operation to the manner in which the capacity of a BAWGN channel varies.  Let noise parameter $Z_n \sim \mathcal{N}(0, 2q\sigma^2_{\text{Total}})$, where $q = \frac{|\textbf{S}_n|\delta}{B}$ is the fraction of the search region measured and $\sigma^2_{\text{Total}} = \frac{B \sigma^2}{2}$ is the half bandwidth variance. Figure~\ref{fig:EN_sigma} show the effects of the half bandwidth variance on the capacity of a search as a function of $q$. Intuitively, the target acquisition rate of the adaptive strategy relates to the time spent searching sets of size $q$ as $q$ varies from $\frac{1}{2}$ to $\frac{\delta}{B}$. This means for sufficiently small $\sigma^2_{\text{Total}}$ ($\leq 0.005$ in this example), the adaptivity gain is negligible since $C_{\text{BAWGN}}(\frac{1}{2}, 2q\sigma^2_{\text{Total}})$ is about 1 for all $q$. For medium range $\sigma^2_{\text{Total}}$ (for e.g., $ 0.05$ in this example), the adaptivity effects the target acquisition rate from $C_{\text{BAWGN}}(\frac{1}{2}, 2q^{\ast}\sigma^2_{\text{Total}})$ to $C_{\text{BAWGN}}(\frac{1}{2}, 2\frac{\delta}{B}\sigma^2_{\text{Total}})$. When $\sigma^2_{\text{Total}}$ grows significantly, however, the capacity drops rather quickly to zero, forcing the non-adaptive strategies to operate close to exhaustive search, whose measurement time increases linearly in $\frac{B}{\delta}$. This is the regime with most significant adaptivity gain as predicted by Corollary~\ref{cor:two_regime_gains}.
\begin{figure}[ht]
     \centering
     \includegraphics[width=0.7\textwidth]{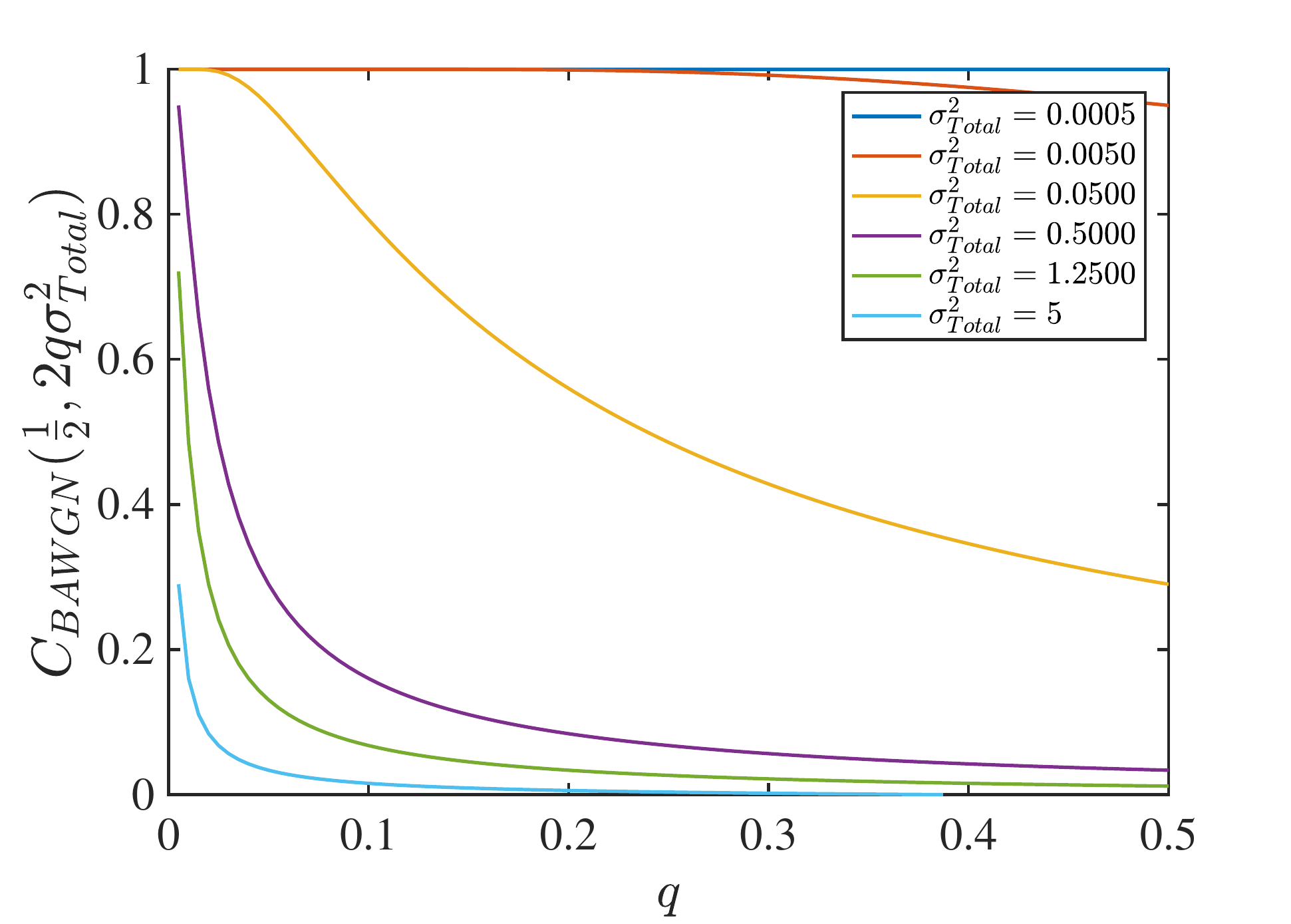}
     \caption{For arbitrary $B$ and $\delta$, and with $\epsilon = 10^{-4}$, $\mathbb{E}_{\mathfrak{c}_{\epsilon}}[\tau]$ as a function of $q$ for different values of total noise variance ($\sigma^2_{Total}$)}
 \label{fig:EN_sigma}
 \end{figure}

%

\subsection{Beyond i.i.d}

In this section, we analyze $\mathbb{E}_{\mathfrak{c}_{\epsilon}}[\tau]$ under a general noise model, as presented in section (VI-A).  Recall, $Y_n \sim \mathcal{N}(X_n, f(|\mbf{S}_n|)\delta \sigma^2)$, where $f$ is a non-decreasing function of the measurement vector $|\textbf{S}_n|$. Figure~\ref{fig:capacity} shows that the behavior of the capacity range of a search with fixed parameters $B$, $\delta$, $\textbf{S}_n$ can be significantly affected by the function $f(\cdot)$. Let us consider the noise function $f(\cdot)$ to be of the form $ |\mbf{S}_n|^{\gamma}$. Figure~\ref{fig:EN_varygamma} shows the plot of dominant terms of the lower bound of Lemma~\ref{lemma:converse_k_1}, and the upper bound of Lemma~\ref{lemma:achv} as a function of $\sigma^2$ for the values of $\gamma \in \{0.5, 1, 2\}$. The adaptivity gain is clearly more significant for larger values of gamma and hence, validates the need for generalizing the noise function. 

 \begin{figure}[ht]
     \centering
     \includegraphics[width=0.7\textwidth]{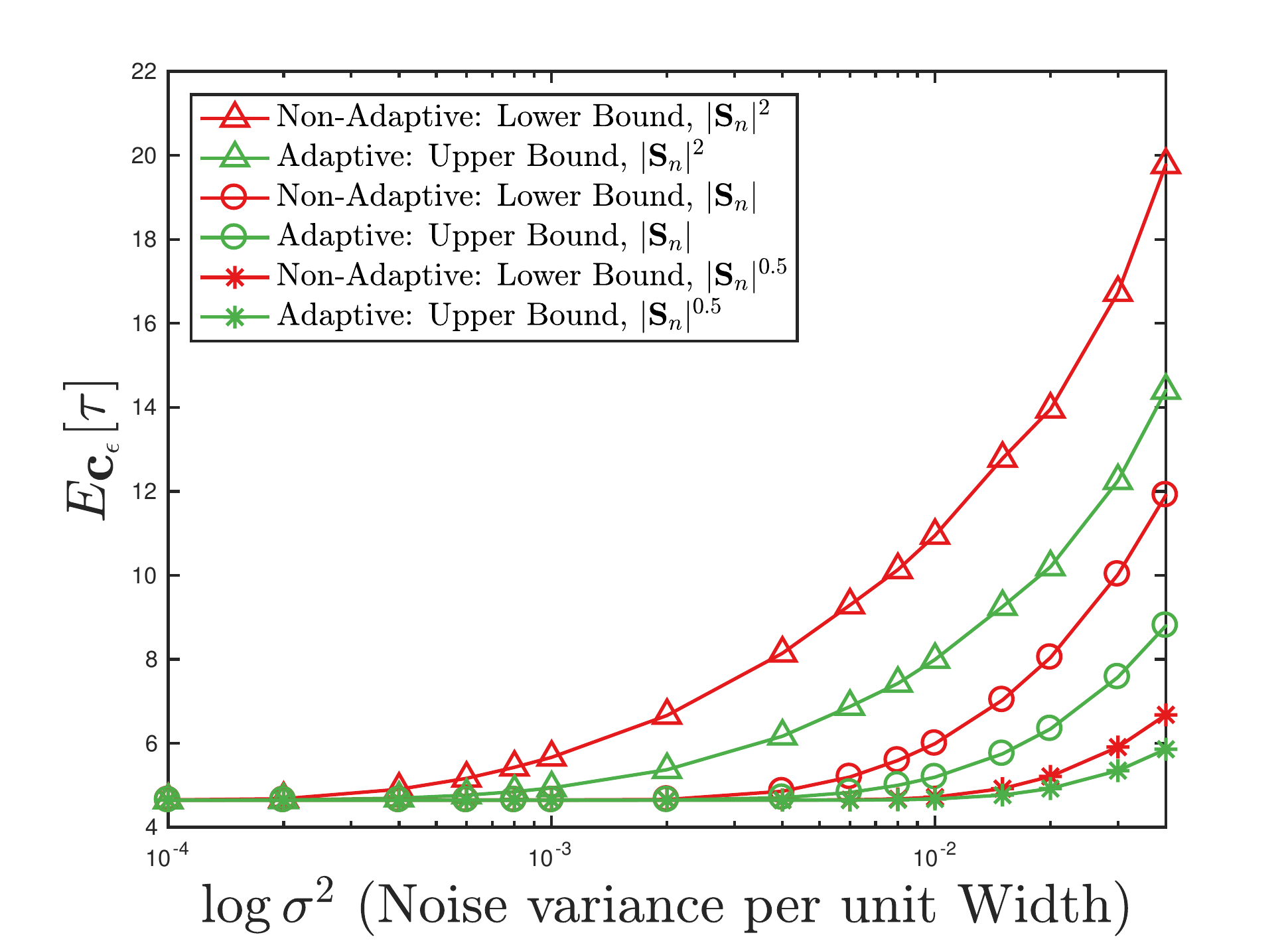}
     \caption{$\mathbb{E}_{\mathfrak{c}_{\epsilon}}[\tau]$ with $\epsilon = 10^{-4}$, $\sigma^2=0.25$ and $B=25$, $\delta =1$,  as a function of $\gamma$ when $Z_n \sim \mathcal{N}(0, |\mbf{S}_n|^{\gamma}\delta \sigma^2)$.} 
 \label{fig:EN_varygamma}
 \end{figure}

\section{Conclusion and Future Work}
We considered the problem of searching for a target's unknown location under measurement dependent Gaussian noise. We showed that this problem is equivalent to channel coding over a BAWGN channel with state and feedback. We used this connection to utilize feedback code based adaptive search strategies. We obtained information theoretic converses to characterize the fundamental limits on the target acquisition rate under both adaptive and non-adaptive strategies. As a corollary, we obtained a lower bound on the adaptivity gain. We identified two asymptotic regimes with practical applications where our analysis shows that adaptive strategies are far more critical when either noise intensity or the total search width is large. In contrast, in scenarios where neither the total width nor noise intensity is large, non-adaptive strategies might perform quite well. The immediate step is the extension of this work to a model with $r>1$ target locations, where the problem has been shown to be equivalent to MAC encoding with feedback~\cite{nancy_asilomar}.

\appendix
\label{appendix}
\subsection{Proof of Lemma~1}
Applying Fano's inequality~\cite{CoverBook2nd} to any non-adaptive search strategy that locates the target among $\frac{B}{\delta}$ locations with $P_e \leq \epsilon$, we have
\begin{align}
\log \left(\frac{B}{\delta}\right)
& \overset{(a)} \leq \frac{1}{1-\epsilon} \sup_{X^n} \sum_{i = 1}^{n} I(X_i, Y_i) +\frac{h(\epsilon)}{1-\epsilon}
\nonumber
\\
&\overset{(b)}\leq \frac{1}{1-\epsilon}\sum_{i = 1}^{n}  C_{\text{BAWGN}} \left(q_i, q_i B \sigma^2 \right) +\frac{h(\epsilon)}{1-\epsilon}
\nonumber
\\
&\leq \frac{n}{1-\epsilon} \max_{q \in \mbf{I}_{\frac{B}{\delta}}} C_{\text{BAWGN}}(q, q B \sigma^2) +\frac{h(\epsilon)}{1-\epsilon},
\end{align}
where $(a)$ follows from the fact that $X_i$ and $Z_i$ for $i = 1,2, \ldots, n$  are independent over time and independent of  past observations due to the non-adaptive nature of the search strategy. Since $X_i \sim \ber(q_i)$ and $Z_i \sim \mathcal{N}(0, q_i  B\sigma^2)$, $(b)$ follows from the fact that $I(X_i, Y_i) = C_{\text{BAWGN}}\left(q_i, q_i B \sigma^2 \right)$. Rearranging the above equation, we have the assertion of the lemma.

\subsection{Proof of Lemma~2}

For any $q \in \mathcal{I}_{\frac{B}{\delta}}$ and under any query vector $\mbf{S}_n \in \mathcal{U}_{\frac{B}{\delta}}$ such that $|\mbf{S}_n| = \frac{qB}{\delta}$ we have the following
\begin{align}
\left| \log \frac{\P(y| \mbf{S}_n, \mbf{W}(i) = 1)}{\P(y| \mbf{S}_n, \mbf{W}(j) = 1)} \right|
=\left\{
\begin{array}{cl}
0 & \text{if $\mbf{S}_n(i) = 1$ and $\mbf{S}_n(j) =1$}, \\
0 & \text{if $\mbf{S}_n(i) \neq 1$ and $\mbf{S}_n(j) \neq 1$},\\
\left| \frac{2y-1}{2qB\sigma^2}\right| &   
\text{Otherwise}.\\
\end{array}
\right.
\end{align}
Hence, we have
\begin{align}
\max_{i,j \in [\frac{B}{\delta}]} \max_{\mbf{S}_n \in \mathcal{U}_{\frac{B}{\delta}}}\int_{-\infty}^{\infty}\P(y| \mbf{S}_n, \mbf{W}(i) = 1)  
\left|\log \frac{\P(y| \mbf{S}_n, \mbf{W}(i) = 1)}{\P(y| \mbf{S}_n, \mbf{W}(j) = 1)} \right|^{1+\gamma}
\nonumber
\\
= 
\max_{q \in \mathcal{I}_{\frac{B}{\delta}}}\left\{\int_{-\infty}^{\infty} \frac{e^{-\frac{y^2}{2qB\sigma^2}}}{\sqrt{2 \pi qB \sigma^2}} \left| \frac{2y-1}{2qB\sigma^2}\right|^{1+\gamma} dy \right\}.
\end{align}
Therefore, there exists $\xi_{\frac{B}{\delta}} < \infty$ and $\gamma > 0$ such that 
\begin{align}
&\max_{i,j \in [\frac{B}{\delta}]} \max_{\mbf{S}_n \in \mathcal{U}_{\frac{B}{\delta}}}\int_{-\infty}^{\infty}\P(y| \mbf{S}_n, \mbf{W}(i) = 1)
\left| \log \frac{\P(y| \mbf{S}_n, \mbf{W}(i) = 1)}{\P(y| \mbf{S}_n, \mbf{W}(j) = 1)} \right|^{1+\gamma}
\leq \xi_{\frac{B}{\delta}}.
\end{align}
Define 
\begin{align}
\psi_{\frac{B}{\delta}}(a) 
: = 
\max_{q \in \mathcal{I}_{\frac{B}{\delta}}}\left\{\int_{-\infty}^{\infty} \frac{e^{-\frac{y^2}{2Bq\sigma^2}}}{\sqrt{2 \pi qB \sigma^2}} \left[ \frac{2y-1}{2qB\sigma^2}\right]_a dy \right\},
\end{align}
and recall that 
\begin{align}
[g]_a = 
\left\{
\begin{array}{ll}
g & \text{if $g \geq a$,}\\
0 & \text{if $g < a$.}
\end{array}
\right.
\end{align}
Note that $\psi_{\frac{B}{\delta}}(a)$ is non-increasing in a, and we have $\psi_{\frac{B}{\delta}}(a) \leq a^{-\gamma}\xi_{\frac{B}{\delta}}$. Hence, $\psi_{\frac{B}{\delta}}(a) \to 0$ as $a \to \infty$.

\subsubsection{Stage I}

\begin{lemma}
\label{lemm:stage_1_time}
Under the fixed composition search strategy while searching over a search region of width $B$ among $\frac{1}{\alpha}$ locations such that $|\mbf{S}^{\prime}_n|\alpha = q^{\ast}$ for $n \geq 1$, the following holds true for all $n \geq 1$ 
\begin{align}
\expe\left[U(\mbs{\rho}^{\prime}_{n+1}) - U(\mbs{\rho}^{\prime}_n)| \mathcal{F}_n , \mbf{S}^{\prime}_n \right] 
\geq C_{\text{BAWGN}}\left(q^{\ast},  q^{\ast} B \sigma^2 \right),
\end{align}
where define $
U(\mbs{\rho}^{\prime}_n)
:= \sum_{i = 1}^{\frac{1}{\alpha}} \mbs{\rho}^{\prime}_n(i)\log \frac{\mbs{\rho}^{\prime}_n(i)}{1-\mbs{\rho}^{\prime}_n(i)}.
$
\end{lemma}

\begin{proof}
The proof follows closely the proof of inequality (9) in~\cite{6400990}. There are $\frac{1}{\alpha}$ locations of length $\alpha B$ and hence query vector $\mbf{S}^{\prime}_n \in \mathcal{U}_{\frac{1}{\alpha}}$. At every time instant under the fixed composition strategy $K^{\ast} = |\mbf{S}_n|= \frac{q^{\ast}}{\alpha}$ number of locations are searched. i.e., a region of length $q^{\ast}B$ is searched. Let $\mathcal{P}_{K^{\ast}}$ denote the collection of all partitions $p$ of search locations $1$ to $\frac{1}{\alpha}$ into sets $A_n^0$ and $A_n^1$ such that $|A_n^1| = K^{\ast}$. The probability of picking a partition $p \in \mathcal{P}_{K^{\ast}}$ is $\lambda_p = {{\frac{1}{\alpha}}\choose{K^{\ast}}}^{-1}$. For simplicity of exposition let $M = \frac{1}{\alpha}$. Also, we have 
$\sum_{p \in \mathcal{P}_{K^{\ast}}}\lambda_p\mbf{1}_{ \{i \in A^0_n\}} = \pi^{\ast}_{0} := \frac{M - K^{\ast}}{M}$, and
$\sum_{p \in \mathcal{P}_{K^{\ast}}}\lambda_p\mbf{1}_{\{i \in A^1_n\}} = \pi^{\ast}_{1} :=\frac{K^{\ast}}{M}$.

Since a region of $q^{\ast}B$ is searched at every time instant, the noise variance is fixed at $q^{\ast}B \sigma^2$. Hence, let $\P_k = \P(Y|X=k, |A^1_n|= K^{\ast}) = \mathcal{N}(k,  q^{\ast} B \sigma^2)$ for $k\in \{0,1 \}$. Consider
\begin{align}
&\expe\left[U(\mbs{\rho}^{\prime}_{n+1}) - U(\mbs{\rho}^{\prime}_n)| \mathcal{F}_n , \mbf{S}_n \right] 
\nonumber
\\
&= 
\sum_{p \in \mathcal{P}_{ K^{\ast}}} \lambda_p \sum_{i = 1}^M \sum_{k = 0}^{1} \mbs{\rho}^{\prime}_n(i) \mbf{1}_{\{i \in A^k_n\}}\kl{ \P_k}{\sum_{j \neq i} \sum_{l = 1}^{1}\frac{\mbs{\rho}^{\prime}_n(j)}{1-\mbs{\rho}^{\prime}_n(i)}\mbf{1}_{\{i \in A^l_n\}} \P_l}
\nonumber
\\
& =  \sum_{i = 1}^M \mbs{\rho}^{\prime}_n(i) \sum_{k = 0}^{1} \pi^{\ast}_{k} \sum_{p \in \mathcal{P}_{ K^{\ast}}} \frac{\lambda_p }{\pi^{\ast}_{k}}  \mbf{1}_{\{i \in A^k_n\}} \kl{\P_k}{\sum_{j \neq i} \sum_{l = 1}^{1}\frac{\mbs{\rho}^{\prime}_n(j)}{1-\mbs{\rho}^{\prime}_n(i)}\mbf{1}_{\{i \in A^l_n\}} \P_l}
\nonumber
\\
&\overset{(a)}\geq
\sum_{i = 1}^M \mbs{\rho}^{\prime}_n(i) \sum_{k = 0}^{1} \pi^{\ast}_{k}  D\left( \P_k \left\| \sum_{j \neq i} \sum_{l = 1}^{1}\frac{\mbs{\rho}^{\prime}_n(j)}{1-\mbs{\rho}^{\prime}_n(i)} \right. \sum_{p \in \mathcal{P}_{ K^{\ast}}} \frac{\lambda_p}{\pi^{\ast}_{k}}\mbf{1}_{\{i \in A^k_t\}}\mbf{1}_{\{i \in A^l_t\}} \P_l\right)
\nonumber
\\
&\overset{(b)}=
\sum_{i = 1}^M \mbs{\rho}^{\prime}_n(i) \left(
 \pi^{\ast}_1\kl{ \P_1 }{ \frac{ K^{\ast}-1}{M-1}\P_1 + \frac{M- K^{\ast}}{M-1} P_0}  \right.
\nonumber
\\
& \hspace{1cm}+
\left. 
\pi^{\ast}_0 \kl{ \P_0}{\frac{M- K^{\ast}-1}{M-1}\P_0 + \frac{ K^{\ast}}{M-1} P_1} 
\right] 
\nonumber
\\
&\geq
\sum_{i = 1}^M \mbs{\rho}^{\prime}_n(i) \left( 
\pi^{\ast}_1 \kl{ \P_1 }{ \frac{ K^{\ast}}{M}\P_1 + \frac{M- K^{\ast}}{M} P_0}
\pi^{\ast}_0 \kl{ \P_0}{\frac{M- K^{\ast}}{M}\P_0 + \frac{ K^{\ast}}{M} P_1}
\right) 
\nonumber
\\
&\overset{(c)} = C_{\text{BAWGN}} \left( q^{\ast}, q^{\ast} B  \sigma^2 \right),
\end{align}
where $(a)$ follows from Jensen's inequality $(b)$ follows from the definition of $\pi_0^{\ast}$, $\pi_1^{\ast}$ and $\sum_{p \in \mathcal{P}_{K^{\ast}}}\lambda_p\mbf{1}_{\{i \in A^0_n\}}\mbf{1}_{\{j \in A^1_n\}} 
= \sum_{p \in \mathcal{P}_{K^{\ast}}}\lambda_p\mbf{1}_{\{i \in A^1_n\}}\mbf{1}_{\{j \in A^0_n\}} = \frac{K^{\ast}(M-K^{\ast})}{M(M-1)}$,
$\sum_{p \in \mathcal{P}_{K^{\ast}}}\lambda_p\mbf{1}_{\{i \in A^0_n\}}\mbf{1}_{\{j \in A^0_n\}} = \frac{\pi^{\ast}_0(M-K^{\ast}-1)}{M-1}$,
$\sum_{p \in \mathcal{P}_{K^{\ast}}}\lambda_p\mbf{1}_{\{i \in A^1_n\}}\mbf{1}_{\{j \in A^1_n\}} = \frac{\pi_1^{\ast}(K^{\ast}-1)}{M-1}$, and $(c)$ is the definition of non-adaptive BAWGN channel capacity with input distribution $\text{Bern}(q^{\ast})$ and noise variance $q^{\ast}B \sigma^2$.
\end{proof}

\begin{lemma}
Under the fixed composition search strategy while searching over a search region of width $B$ among $\frac{1}{\alpha}$ locations such that $|\mbf{S}^{\prime}_n|\alpha = q^{\ast}$ for $n \geq 1$, the following holds true for the expected number of queries while searching with $\P_e \leq \frac{\epsilon}{2}$
\begin{align}
\expe_{\mathfrak{c}_{\epsilon}^1}[\tau^1]
\leq 
\frac{\log \frac{1}{\alpha} + \log \frac{2}{\epsilon} + \log \log \frac{B}{\delta} + a_{\eta}}{C_{\text{BAWGN}}\left(q^{\ast}, q^{\ast} B \sigma^2 \right) - 
\eta}.
\end{align}
\end{lemma}
Proof is similar to the proof of Lemma~\ref{lemm:sortPM_tau}.

\subsubsection{Stage II}
Note that BAWGN capacity for all $q \in \mathcal{I}_{\frac{B}{\delta}}$
\begin{align}
C_{\text{BAWGN}}\left( \frac{1}{2}, qB\sigma^2\right)
&= 
D\left( \mathcal{N}(0, qB\sigma^2) \left\| \frac{1}{2}\mathcal{N}(0,  qB \sigma^2) + \frac{1}{2}\mathcal{N}(1,  qB\sigma^2)\right. \right)
\nonumber
\\
&=
D\left(\mathcal{N}(1, qB \sigma^2) \left\| \frac{1}{2}\mathcal{N}(0, qB\sigma^2) + \frac{1}{2}\mathcal{N}(1, qB\sigma^2)\right. \right).
\end{align}
Hence, the following Lemma follows from Proposition~3 in~\cite{7031961}.
\begin{lemma}
\label{lemm:EJSGaussian}
Under the sortPM search strategy while searching over a search region of width $\alpha B$ among $\frac{\alpha B}{\delta}$ locations, the following holds true for all $n \geq 1$ 
\begin{align}
\expe\left[U(\mbs{\rho}^{\prime \prime}_{n+1}) - U(\mbs{\rho}^{\prime \prime}_n)| \mathcal{F}_n , \mbf{S}_n \right] 
\geq C_{\text{BAWGN}}\left(\frac{1}{2}, \frac{\alpha B \sigma^2}{2} \right),
\end{align}
where define 
$
U(\mbs{\rho}^{\prime \prime}_n)
:= \sum_{i = 1}^{\frac{\alpha B}{\delta}} \mbs{\rho}^{\prime \prime}_n(i)\log \frac{\mbs{\rho}^{\prime \prime}_n(i)}{1-\mbs{\rho}^{\prime \prime}_n(i)}.
$
\end{lemma}

\begin{lemma}
\label{lemm:sortPM_tau}
Under the sortPM search strategy, the following holds true for the expected number of queries while searching over the search width $\alpha B$ among $\frac{\alpha B}{\delta}$ locations with $\P_e \leq \frac{\epsilon}{2}$
\begin{align}
\expe_{\mathfrak{c}_{\epsilon}^2}[\tau^{2}]
\leq 
\frac{\log \frac{\alpha B}{\delta} + \log \frac{2}{\epsilon} + \log \log \frac{\alpha B}{\delta} + a_{\eta}}{C_{\text{BAWGN}}\left(\frac{1}{2}, \frac{\alpha B \sigma^2}{2} \right) - 
\eta},
\end{align}
where $a_{\eta}$ is the solution of the following equation
\begin{align}
\eta =\frac{a}{a-3}\psi_{\frac{B}{\delta}}(a-3).
\end{align}
\end{lemma}

\begin{proof}
Let $M = \frac{\alpha B}{\delta}$. Let $\tilde{\rho}^{\prime} = 1 - \frac{1}{1+\max\{\log M, \frac{2}{\epsilon}\}}$. Now, define $U^{\prime}(\mbs{\rho}^{\prime \prime}_0) = U(\mbs{\rho}^{\prime \prime}_0)- \log \frac{\tilde{\rho}^{\prime}}{1-\tilde{\rho}^{\prime}}$ and define $U^{\prime}(\mbs{\rho}^{\prime \prime}_n)$ as follows: if $U^{\prime}(\mbs{\rho}^{\prime \prime}_n) < 0$, then
\begin{align}
\label{eq:UPrimeCase1}
U^{\prime}(\mbs{\rho}^{\prime \prime}_{n+1})
= 
\left\{
\begin{array}{ll}
U(\mbs{\rho}^{\prime \prime}_{n+1}) - U(\mbs{\rho}^{\prime \prime}_{n}) + U^{\prime}(\mbs{\rho}^{\prime \prime}_n) & \text{if $U(\mbs{\rho}^{\prime \prime}_{n+1}) - U(\mbs{\rho}^{\prime \prime}_{n}) < a - U^{\prime}(\mbs{\rho}^{\prime \prime}_n)$,} \\
a & \text{if $U(\mbs{\rho}^{\prime \prime}_{n+1}) - U(\mbs{\rho}^{\prime \prime}_{n}) \geq a - U^{\prime}(\mbs{\rho}^{\prime \prime}_n)$,}
\end{array}
\right.
\end{align}
and if $U^{\prime}(\mbs{\rho}^{\prime \prime}_n) \geq 0$, then 
\begin{align}
\label{eq:UPrimeCase2}
U^{\prime}(\mbs{\rho}^{\prime \prime}_{n+1})
= 
\left\{
\begin{array}{ll}
U(\mbs{\rho}^{\prime \prime}_{n+1}) - U(\mbs{\rho}^{\prime \prime}_{n}) + U^{\prime}(\mbs{\rho}^{\prime \prime}_n) &
 \text{if $U(\mbs{\rho}^{\prime \prime}_{n+1}) - U(\mbs{\rho}^{\prime \prime}_{n}) < a$,} \\
a + U^{\prime}(\mbs{\rho}^{\prime \prime}_n)
& \text{if $U(\mbs{\rho}^{\prime \prime}_{n+1}) - U(\mbs{\rho}^{\prime \prime}_{n}) \geq a $.}
\end{array}
\right.
\end{align}
By induction we can show that
\begin{align}
\label{eq:UPrimeAndU}
\log \frac{\tilde{\rho}^{\prime}}{1-\tilde{\rho}^{\prime}}
\leq U(\mbs{\rho}^{\prime \prime}_n) - U^{\prime}(\mbs{\rho}^{\prime \prime}_n).
\end{align}
We have
\begin{align}
&\expe\left[U^{\prime}(\mbs{\rho}^{\prime \prime}_{n+1}) - U^{\prime}(\mbs{\rho}^{\prime \prime}_n)| \mathcal{F}_n \right]
\nonumber
\\
&= 
\expe\left[U(\mbs{\rho}^{\prime \prime}_{n+1}) - U(\mbs{\rho}^{\prime \prime}_n)| \mathcal{F}_n \right]
+
\expe\left[ \left[-b - U(\mbs{\rho}^{\prime \prime}_{n+1}) + U(\mbs{\rho}^{\prime \prime}_n)  -U^{\prime}(\mbs{\rho}^{\prime \prime}_n) \mbf{1}_{\{U^{\prime}(\mbs{\rho}^{\prime \prime}_n) < 0 \}} \right]^{+}| \mathcal{F}_n\right]
\nonumber
\\
& \overset{(a)}\geq \expe\left[U(\mbs{\rho}^{\prime \prime}_{n+1}) - U(\mbs{\rho}^{\prime \prime}_n)| \mathcal{F}_n \right] - 
\frac{a}{a-3}\psi_{\frac{B}{\delta}}(a-3)
\nonumber
\\
& \overset{(b)}\geq C_{\text{BAWGN}}\left(\frac{1}{2}, \frac{\alpha B \sigma^2}{2} \right) - 
\frac{a}{a-3}\psi_{\frac{B}{\delta}}(a-3),
\end{align}
where $(a)$ follows from MN thesis eq (4.140) and $(b)$ follows Lemma~\ref{lemm:EJSGaussian}. Let $\tau^{\prime} = \min\{n: U^{\prime}(\mbs{\rho}^{\prime \prime}_n) \geq 0\}$ and $\tau_{\frac{\epsilon}{\epsilon}} = \min\{n: U(\mbs{\rho}^{\prime \prime}_n) \geq \log \frac{\tilde{\rho}}{1-\tilde{\rho}}\}$ where $\tilde{\rho} = 1- \frac{2}{\epsilon}$. From equation~\eqref{eq:UPrimeAndU} and since $\tilde{\rho}^{\prime} > \tilde{\rho}$, we have
\begin{align}
\label{eq:stopping_time_comp}
\expe_{\mathfrak{c}_{\epsilon}^2}[\tau_{\frac{\epsilon}{2}}] \leq \expe_{\mathfrak{c}_{\epsilon}^2}|\tilde{\tau}^{\prime}].
\end{align}
The sequence $\frac{U^{\prime}(\mbs{\rho}^{\prime \prime}_n)}{C_{\text{BAWGN}}\left(\frac{1}{2}, \frac{\alpha B \sigma^2}{2} \right) - 
\frac{a}{a-3}\psi_{\frac{B}{\delta}}(a-3)} - n$ forms a submartingale with respect to filtration $\mathcal{F}_n$. Now by Doob's Stopping Theorem we have
\begin{align}
\frac{U^{\prime}(\mbs{\rho}^{\prime \prime}_{0})}{C_{\text{BAWGN}}\left(\frac{1}{2}, \frac{\alpha B \sigma^2}{2} \right) - 
\frac{a}{a-3}\psi_{\frac{B}{\delta}}(a-3)}
\leq
\expe \left[
\frac{U^{\prime}(\mbs{\rho}^{\prime \prime}_{\tilde{\tau}^{\prime}})}{C_{\text{BAWGN}}\left(\frac{1}{2}, \frac{\alpha B \sigma^2}{2} \right) - 
\frac{a}{a-3}\psi_{\frac{B}{\delta}}(a-3)}
- \tilde{\tau}^{\prime}
\right].
\end{align}
Hence, we have
\begin{align}
\label{eq:doon_ineq}
\expe_{\mathfrak{c}_{\epsilon}^2}[\tilde{\tau}^{\prime}] 
&\leq 
\frac{-U^{\prime}(\mbs{\rho}^{\prime \prime}_0) + \expe[U^{\prime}(\mbs{\rho}^{\prime \prime}_{\tilde{\tau}^{\prime}})]}{C_{\text{BAWGN}}\left(\frac{1}{2}, \frac{\alpha B \sigma^2}{2} \right) - 
\frac{a}{a-3}\psi_{\frac{B}{\delta}}(a-3)}
\nonumber
\\
&
= \frac{-U(\mbs{\rho}^{\prime \prime}_0) + \log \frac{\tilde{\rho}^{\prime}}{1-\tilde{\rho}^{\prime}} + \expe[U^{\prime}(\mbs{\rho}^{\prime \prime}_{\tilde{\tau}^{\prime}})]}{C_{\text{BAWGN}}\left(\frac{1}{2}, \frac{\alpha B \sigma^2}{2} \right) - 
\frac{a}{a-3}\psi_{\frac{B}{\delta}}(a-3)}
\nonumber
\\
& \overset{(a)}\leq 
\frac{\log \frac{\alpha B}{\delta} +\log \log \frac{\alpha B}{\delta} + \log \frac{2}{\epsilon} + \expe[U^{\prime}(\mbs{\rho}^{\prime \prime}_{\tilde{\tau}^{\prime}})]}{C_{\text{BAWGN}}\left(\frac{1}{2}, \frac{\alpha B \sigma^2}{2} \right) - 
\frac{a}{a-3}\psi_{\frac{B}{\delta}}(a-3)}
\nonumber
\\
& \overset{(b)}\leq 
\frac{\log \frac{\alpha B}{\delta} + \log \log \frac{\alpha B}{\delta} + \log \frac{2}{\epsilon} + a}{C_{\text{BAWGN}}\left(\frac{1}{2}, \frac{\alpha B \sigma^2}{2} \right) - 
\frac{a}{a-3}\psi_{\frac{B}{\delta}}(a-3)},
\end{align}
where $(a)$ follows from the fact that $U(\mbs{\rho}^{\prime \prime}_0) = -\log (\frac{B}{\delta} - 1)$ and $(b)$ follows from the fact that for all $n < \tau^{\prime}$, $U^{\prime}(\mbs{\rho}^{\prime \prime}_n) < 0$ and hence from equation~\eqref{eq:UPrimeCase1} we have $U^{\prime}(\mbs{\rho}^{\prime \prime}_{\tilde{\tau}^{\prime}}) < a$. Let $\eta> 0$ such that $\eta 
\ll C_{\text{BAWGN}}\left(\frac{1}{2}, \frac{\alpha B \sigma^2}{2} \right)$. Choose $a_{\eta}$ such that 
\begin{align}
\eta =\frac{a}{a-3}\psi_{\frac{B}{\delta}}(a-3).
\end{align}
We have the assertion of the lemma by combining above equation with equations~\eqref{eq:stopping_time_comp} and~\eqref{eq:doon_ineq}.
\end{proof}

\subsection{Proof of Corollary~2}
Choose $a_{\eta} = \log \log \frac{B}{\delta}$ so that $\eta$ goes to zero as $\frac{B}{\delta} \to \infty$, and choose $\alpha(\frac{B}{\delta}) = \frac{1}{\log \frac{B}{\delta}}$. Note that $\alpha (\frac{B}{\delta})$ goes to $0$ slower than $\delta$ goes to $0$. Combining this with Theorem~\ref{thm:gain_lower_bound} and using the fact $\lim_{\delta \to 0}C_{\text{BAWGN}}\left( \frac{1}{2}, \frac{1}{2} \alpha\left( \frac{B}{\delta}\right) B \sigma^2 \right) = 1$, we have equation~(\ref{eq:delta_gain}). Similarly, note that $\alpha (\frac{B}{\delta})$  goes to $0$ slower than $B$ goes to $\infty$. Using loose approximations $C_{\text{BAWGN}}(q^{\ast}, q^{\ast}B\sigma^2)\leq \frac{\log e}{B \sigma^2}$ and $C_{\text{BAWGN}}\left(\frac{1}{2}, \alpha(\frac{B}{\delta}) B \sigma^2 \right) \geq \frac{\log(\frac{B}{\delta})}{16B \sigma^2} \left( 1 - \frac{\log(\frac{B}{\delta})}{16B \sigma^2} \right)$ with Theorem~\ref{thm:gain_lower_bound} we have equations~(\ref{eq:B_NA}--\ref{eq:B_gain}).

%
%
%
%
%
%
%

\color{black}
\bibliographystyle{IEEEtran}
\bibliography{HypTest}


\end{document}